\documentclass[10pt,conference,compsocconf]{IEEEtran}

\usepackage{url}
\usepackage{color}
\usepackage{amsmath}
\usepackage{amssymb}
\usepackage{mathpartir} 
\usepackage{enumerate}
\usepackage{mdwlist}
\usepackage{makeidx}
\usepackage{graphics,graphicx}

\makeindex

\makeatletter
\def\url@leostyle{%
  \@ifundefined{selectfont}{\def\UrlFont{\sf}}{\def\UrlFont{\small\ttfamily}}}
\makeatother
\newtheorem{definition}{Definition}
\newtheorem{theorem}{Theorem}

\newcommand{\wsep}{{\,,\,}}

\newcommand{\ninfty}{-\infty}

\newcommand{\btrue}{\text{$\tt true$}}

\newcommand{\astep}{\mathrel{\mapsto}}
\newcommand{\cstep}{\mathrel{\longrightarrow}}
\newcommand{\cstept}[1]{\mathrel{\stackrel{#1}{\cstep}}}
\newcommand{\tstep}{\mathrel{\hookrightarrow}}
\newcommand{\tstept}[1]{\mathrel{\stackrel{#1}{\tstep}}}

\newcommand{\rname}[1]{\ensuremath{\text{#1}}}

\newcommand{\Server}{\ensuremath{{\sf Server}}}

\newcommand{\User}{\ensuremath{{\sf User}}}

\newcommand{\Notary}{\ensuremath{{\sf Notary}}}
\newcommand{\Adversary}{\ensuremath{{\sf Adversary}}}

\newcommand{\action}[1]{\texttt{#1}}

\newcommand{\pred}[1]{{\tt #1}}

\newcommand{\at}{\mathrel{\mathrm{@}}}

\newcommand{\imp}{\mathrel{\supset}}

\newcommand{\cconj}{\mathrel{\wedge}}

\newcommand{\cut}[1]{}

\newcommand{\redlabel}{\ensuremath{r}}

\title{Program Actions as Actual Causes: \\ A Building Block for Accountability }
\author{
\IEEEauthorblockN{Anupam Datta}
\IEEEauthorblockA{$\quad$Carnegie Mellon University$\quad$\\
Email: danupam@cmu.edu}
\and
\IEEEauthorblockN{Deepak Garg}
\IEEEauthorblockA{$\quad$Max Planck Institute for Software Systems$\quad$\\
Email: dg@mpi-sws.org}
\and
\IEEEauthorblockN{Dilsun Kaynar}
\IEEEauthorblockA{$\quad$Carnegie Mellon University$\quad$\\
Email: dilsunk@cmu.edu}
\and
\IEEEauthorblockN{$\qquad \qquad \quad$}
\IEEEauthorblockA{$$}
\and
\IEEEauthorblockN{Divya Sharma}
\IEEEauthorblockA{Carnegie Mellon University\\
Email: divyasharma@cmu.edu}
\and
\IEEEauthorblockN{Arunesh Sinha}
\IEEEauthorblockA{University of Southern California \\
Email: aruneshs@usc.edu}
}

\begin{document}

\maketitle

\begin{abstract}
Protocols for tasks such as authentication, electronic voting, and
secure multiparty computation ensure desirable security properties if
agents follow their prescribed programs. However, if some agents
deviate from their prescribed programs and a security property is
violated, it is important to hold agents \emph{accountable} by
determining which deviations actually caused the violation.  Motivated
by these applications, we initiate a formal study of \emph{program actions as
actual causes}.  Specifically, we define in an interacting program model 
what it means for a set of program actions to be an actual cause of a violation. 
We present a sound technique for establishing program actions as
actual causes.  We demonstrate the value of this formalism in two
ways. First, we prove that violations of a specific class of safety
properties always have an actual cause. Thus, our definition applies
to relevant security properties.  Second, we provide a cause analysis
of a representative protocol designed to address weaknesses in the
current public key certification infrastructure.
\end{abstract}


\IEEEpeerreviewmaketitle

\section{Introduction} \label{sec:introduction} 
Ensuring accountability for security violations is essential in a wide
range of settings. For example, protocols for authentication and key
exchange~\cite{kaufman1995}, electronic voting~\cite{rivest2007},
auctions~\cite{parkes2008}, and secure multiparty computation (in the
semi-honest model)~\cite{goldreich1987} ensure desirable security
properties if protocol parties follow their prescribed
programs. However, if they deviate from their prescribed programs and
a security property is violated, determining which agents should be
held accountable and appropriately punished is important to deter
agents from committing future violations.  Indeed the importance of
accountability in information systems has been recognized in prior
work
\cite{nissenbaum1996,lampson2004,backesDDMT06,haeberlenKD07,weitzner2008,jagadeesanJPR09,kusters2010}.
Our thesis is
that \emph{actual causation} (i.e., identifying which agents' actions caused a specific violation)
is a useful building block for accountability in decentralized multi-agent systems, including but
not limited to security protocols and ceremonies~\cite{ellisonceremonydesign}. 

Causation has been of interest to philosophers and ideas from
philosophical literature have been introduced into computer science by
the seminal work of Halpern and Pearl~\cite{halpernpearl2001,
halpernpearl2005, pearlbook2000}. In particular, counterfactual
reasoning is appealing as a basis for study of causation. 
\cut{
Indeed, actual causation is a building block for causal explanations~\cite{halpern2005explanations}, which 
can used to provide an account of why a violation happened. It also is a building block for 
blame assignment in influential theories of moral and legal blame~\cite{moore2009,hart1985,shaver2012attribution, alicke2000culpable,
feinberg1985suaculpa, kenner1967blaming}.

Actual causation has been studied extensively in philosophy, law, and
computer science where }
Much of the definitional activity has centered
around the question of what it means for event $c$ to be an actual
cause of event $e$.  An answer to this question is useful to arrive at
causal judgments for specific scenarios such as ``John's smoking
causes John's cancer'' rather than general inferences such as
``smoking causes cancer'' (The latter form of judgments are studied
in the related topic of type causation~\cite{pearlbook2000}).
Notably, Hume~\cite{hume1748} identified actual causation with
counterfactual dependence---the idea that $c$ is an actual cause of
$e$ if had $c$ not occurred then $e$ would not have occurred. While
this simple idea does not work if there are independent causes, the
counterfactual interpretation of actual causation has been developed
further and formalized in a number of influential works (see, for
example,~\cite{lewis1973,pearlbook2000,hallbook,mackie1965,wright1985,halpernpearl2005}).

Even though
applications of counterfactual causal analysis are starting to emerge
in the fields of AI, model-checking, and programming languages,
causation has not yet been studied in connection with security
protocols and violations thereof. On the other hand, causal analysis
seems to be an intuitive building block for answering some very
natural questions that have direct relevance to accountability such as (i) \emph{why} a particular violation occurred, (ii) \emph{what} component in the protocol is blameworthy
for the violation and (iii) \emph{how} the protocol could have been
designed differently to preempt violations of this sort. Answering
these questions requires an in-depth study of, respectively,
explanations, blame-assignment, and protocol design, which are
interesting problems in their own right, but are not the explicit
focus of this paper. Instead, we focus on a formal definition
of causation that we believe formal studies of these problems will
need. Roughly speaking, explanations can be used to
provide an \emph{account} of the violation, \emph{blame assignment}
can be used to hold agents \emph{accountable} for the violation,
and protocol design informed by these would lead to protocols
with better accountability guarantees. We further elaborate on explanations and blame-assignment in Section~\ref{sec:domains}. 

Formalizing actual causes as a building block for accountability in
decentralized multi-agent systems raises new conceptual and technical
challenges beyond those addressed in the literature on events as
actual causes. In particular, prior work does not account for the
program dynamics that arise in such settings. Let us consider a
simple protocol example.  In the
movie \emph{Flight}~\cite{flight2012}, a pilot drinks and snorts
cocaine before flying a commercial plane, and the plane goes into a
locked dive in mid-flight. While the pilot's behavior is found to be
deviant in this case---he does not follow the prescribed protocol
(program) for pilots---it is found to not be an actual cause of the
plane's dive. The actual cause was a deviant behavior by the
maintenance staff---they did not replace a mechanical component that
should have been replaced. Ideally, the maintenance staff should have
inspected the plane prior to take-off according to their prescribed
protocol.

This example is useful to illustrate several key ideas that influence
the formal development in this paper.  First, it illustrates the
importance of capturing the \emph{actual interactions} among agents in
a decentralized multi-agent system with non-deterministic execution
semantics.  The events in the movie could have unfolded in a different
order but it is clear that the actual cause determination needs to be
done based on the sequence of events that happened in reality.  For
example, had the maintenance staff replaced the faulty
component \emph{before} the take-off the plane would not have gone
into a dive.  Second, the example motivates us to hold accountable
agents who exercise their choice to execute a deviant \emph{program} that
actually caused a violation. The maintenance staff had the choice to
replace the faulty component or not where the task of replacing the
component could consist of multiple steps. It is important to identify
which of those steps were crucial for the occurrence of the dive. Thus, we
focus on formalizing \emph{program actions} executed in sequence (by agents)
as actual causes of violations rather than individual, independent
events as formalized in prior work.  Finally, the example highlights
the difference between deviance and actual causes---a difference also
noted in prior work on actual causation. This difference is important
from the standpoint of accountability. In particular, the punishment
for deviating from the prescribed protocol could be suspension or
license revocation whereas the punishment for actually causing a plane
crash in which people died could be significantly higher (e.g.,
imprisonment for manslaughter). The first and second ideas, reflecting
our program-based treatment, are the most significant points of
difference from prior work on actual
causation~\cite{halpernpearl2005,defaults2008} while the third idea is
a significant point of difference from prior work in
accountability~\cite{feigenbaum2011,
feigenbaum2011deterrence, kusters2010, backesDDMT06}.

The central contribution of this paper is a formal definition of 
\emph{program actions as actual causes}.  Specifically, we
define what it means for a set of program actions to be an actual
cause of a violation. The definition considers a set of interacting
programs whose concurrent execution, as recorded in a log, violates a
trace property. It identifies a subset of actions (program steps) of
these programs as an actual cause of the violation. The definition
applies in two phases. The first phase identifies what we
call \emph{Lamport causes}. A Lamport cause is a minimal prefix of the
log of a violating trace that can account for the violation. In the
second phase, we refine the actions on this log by removing the
actions which are merely
\emph{progress enablers} and obtain \emph{actual action causes}. The former
contribute only indirectly to the cause by enabling the actual action
causes to make progress; the exact values returned by progress
enabling actions are irrelevant.

We demonstrate the value of this formalism in two ways. First, we
prove that violations of a precisely defined class of safety
properties always have an actual cause. Thus, our definition applies
to relevant security properties.  Second, we provide a cause analysis
of a representative protocol designed to address weaknesses in the
current public key certification infrastructure. Moreover, our example
illustrates that our definition cleanly handles the separation between
joint and independent causes --a recognized challenge for actual cause
definitions~\cite{halpernpearl2001,
halpernpearl2005, pearlbook2000}. 

 In addition, we discuss how
this formalism can serve as a building block for causal explanations
and exoneration (i.e., soundly identifying agents who should not be
blamed for a violation). We leave the technical development of these
concepts for future work.

The rest of the paper is organized as follows. Section~\ref{sec:example} describes a representative example
which we use throughout the paper to explain important concepts.
Section~\ref{sec:definitions} gives formal definitions for program actions as
actual causes of security violations. We apply the causal analysis to
the running example in Section~\ref{sec:application}. We discuss the use of our causal analysis techniques for providing explanations and assigning blame in Section~\ref{sec:domains}. We survey
additional related work in Section~\ref{sec:related} and conclude in
Section~\ref{sec:conclusion}.


\section{Motivating example} \label{sec:example}

In this section we describe an example protocol designed to increase
accountability in the current public key infrastructure. We use the
protocol later to illustrate key concepts in defining causality.

\paragraph{Security protocol}  
Consider an authentication protocol in which a user (\User1)
authenticates to a server (\Server1) using a pre-shared password over
an adversarial network. \User1\ sends its user-id to \Server1\ and
obtains a public key signed by
\Server1. However, \User1\ would need
inputs from additional sources when \Server1\ sends its public key for
the first time in a protocol session to verify that the key is indeed
bound to \Server1's identity. In particular, \User1\ can verify the
key by contacting multiple notaries in the spirit
of \emph{Perspectives}~\cite{wendlandt2008}.  For simplicity, we
assume \User1\ verifies \Server1's public key with three authorized
notaries---\Notary1, \Notary2, \Notary3---and accepts the key if and
only if the majority of the notaries say that the key is
legitimate. To illustrate some of our ideas, we also consider a
parallel protocol where two parties (\User2 and \User3) communicate
with each other.

We assume that the prescribed programs
for \Server1, \User1, \Notary1, \Notary2, \Notary3, \User2 and \User3
impose the following requirements on their behavior: (i) \Server1\
stores \User1's password in a hashed form in a secure private memory
location.  (ii)
\User1\ requests access to the account by sending
an encryption of the password
(along with its identity and a timestamp) to \Server1\ after verifying
\Server1's public key with a majority of the notaries.  (iii) The notaries retrieve the key from
their databases and attest the key correctly. (iv) \Server1\ decrypts
and computes the hashed value of the password. (v) \Server1\ matches the computed hash value with the previously stored value in
the memory location when the account was first created; if the two
hash values match, then \Server1\ grants access to the account
to \User1. (vi) In parallel, \User2 generates and sends a nonce
to \User3. (vii) \User3 generates a nonce and responds to \User2.

\paragraph{Security property}
The prescribed programs in our example aim to achieve the property
that only the user who created the account and password (in this case, \User1) gains access
to the account. 

\paragraph{Compromised Notaries Attack}
We describe an attack scenario and use it to illustrate nuances in
formalizing program actions as actual causes. \User1\
executes its prescribed program. \User1\ sends an access
request to \Server1. An \Adversary\ intercepts the message and sends a
public key to \User1\ pretending to be \Server1. \User1\ checks
with \Notary1, \Notary2\ and \Notary3\ who falsely verify \Adversary's
public key to be \Server1's key. Consequently, \User1\ sends the
password to \Adversary. \Adversary\ then initiates a protocol with
{\Server1} and gains access to \User1's account. In parallel, \User2\
sends a request to \Server1 and receives a response
from \Server1. Following this interaction, \User2 forwards the message
to \User3.   We assume
that the actions of the parties are recorded on a \emph{log}, say
$l$. Note that this log contains a violation of the security property
described above since \Adversary\ gains access to an account owned
by \User1.

First, our definition finds \emph{program actions as causes} of
violations.  At a high-level, as mentioned in the introduction, our
definition applies in two phases. The first phase
(Section~\ref{sec:definitions}, Definition~\ref{definition:cause1'})
identifies a minimal prefix (Phase~1, \emph{minimality}) of the log
that can account for the violation i.e. we consider all scenarios
where the sequence of actions execute in the same order as on the log,
and test whether it suffices to recreate the violation in the absence
of all other actions (Phase~1, \emph{sufficiency}).  In
our example, this first phase will output a minimal prefix of log $l$
above.  In this case, the minimal prefix will not contain interactions
between \User2 and \User3 after \Server1 has granted access to
the \Adversary\ (the remaining prefix will still contain a
violation).

Second, a
nuance in defining the notion of \emph{sufficiency} (Phase~1,
Definition~\ref{definition:cause1'}) is to constrain the interactions
which are a part of the actual cause set in a manner that is
consistent with the interaction recorded on the log.  This constraint
on interactions is quite subtle to define and depends on how strong a
coupling we find appropriate between the log and possible
counterfactual traces in sufficiency: if the constraint is too weak
then the violation does not reappear in all sequences, thus missing
certain causes; if it is too strong it leads to counter-intuitive
cause determinations. For example, a weak notion of consistency is to
require that each program locally execute the same prefix in
sufficiency as it does on the log i.e. consistency w.r.t. program
actions for individual programs. This notion does not work because for
some violations to occur the \emph{order of interactions} on the log
among programs is important.  A notion that is too strong is to
require matching of the total order of execution of all actions across
all programs.  We present a formal notion of \emph{consistency} by
comparing log projections (Section~\ref{sec:support-defs}) that
balance these competing concerns.

Third, note that while Phase~1 captures a minimal prefix of the log sufficient for the violation, it might be possible to remove actions from this prefix which are merely required for a program execution to progress. For instance note that while all three notaries' actions are required for
\User1\ to progress (otherwise it would be stuck waiting
to receive a message) and the violation to occur, the actual message
sent by one of the notaries is irrelevant since it does not affect the
majority decision in this example. Thus, separating out actions which
are \emph{progress enablers} from those which provide information that
causes the violation is useful for fine-grained causal
determination. This observation motivates the final piece (Phase~2) of
our formal definition (Definition~\ref{definition:cause2a}).

Finally, notice that in this
example \Adversary, \Notary1, \Notary2, \Notary3, \Server1 and \User2
deviate from the protocol described above. However, the deviant
programs are not sufficient for the violation to occur without the
involvement of \User1, which is also a part of the causal set. We thus
seek a notion of sufficiency in defining a set of programs as a joint
actual cause for the violation.  Joint causation is also significant
in legal contexts~\cite{european2005principles}. For instance, it is
useful for holding liable a group of agents working together when none
of them satisfy the cause criteria individually but together their
actions are found to be a cause. The ability to distinguish between
joint and independent (i.e., different sets of programs that
independently caused the violation) causes is an important criterion
that we want our definition to satisfy. In particular, Phase~2 of our
definition helps identify independent causes. For instance, in our
example, we get three different independent causes depending on which
notary's action is treated as a progress enabler.  Our ultimate goal
is to use the notion of actual cause as a building block for
accountability --- the independent vs. joint cause distinction is
significant when making deliberations about accountability and
punishment for liable parties.  We can use the result of our causal determinations to further remove
deviants whose actions are required for the violation to occur but
might not be blameworthy
(Section~\ref{sec:domains}).


\section{Actual Cause Definition}\label{sec:definitions} 
We present our language model in Section~\ref{sec:model}, auxiliary
notions in Section~\ref{sec:support-defs}, properties of interest to
our analysis in Section~\ref{sec:properties}, and the formal
definition of program actions as actual causes in
Section~\ref{sec:cause1}.



\subsection{Model} 
\label{sec:model} 

We model programs in a simple concurrent language, which we call
$L$. The language contains sequential expressions, $e$, that execute
concurrently in threads and communicate with each other
through \action{send} and \action{recv} commands. Terms, $t$, denote
messages that may be passed through expressions or across
threads. Variables $x$ range over terms. An expression is a sequence
of actions, $\alpha$. An action may do one of the following: execute a
primitive function $\zeta$ on a term $t$ (written $\zeta(t)$), or send
or receive a message to another thread (written $\action{send}(t)$ and
$\action{recv}()$, respectively). We also include very primitive
condition checking in the form of $\action{assert}(t)$.
\[\begin{array}{lllll}
\mbox{Terms} & t & ::= & x ~|~ \ldots\\
 
\mbox{Actions} & \alpha & ::= & 
\zeta(t) ~|~ \action{send}(t) 
~|~ \action{recv}()\\

\mbox{Expressions}  & e  & ::= & t ~|~ (b: x = \alpha); e_2 ~|~ \action{assert}(t); e

\end{array}\]

Each action $\alpha$ is labeled with a unique line number, written
$b$. Line numbers help define traces later. We omit line numbers when
they are irrelevant. Every action and expression in the language
evaluates to a term and potentially has side-effects. The term
returned by action $\alpha$ is bound to $x$ in evaluating $e_2$ in the
expression $(b: x= \alpha); e_2$.

Following standard models of protocols, $\action{send}$ and
$\action{recv}$ are untargeted in the operational semantics: A message
sent by a thread may be received by any thread. Targeted communication
may be layered on this basic semantics using cryptography. For
readability in examples, we provide an additional first argument to
$\action{send}$ and $\action{recv}$ that specifies the \emph{intended}
target (the operational semantics ignore this intended target). Action
$\action{send}(t)$ always returns $0$ to its continuation.

Primitive functions $\zeta$ model thread-local computation like
arithmetic and cryptographic operations. Primitive functions can also
read and update a \emph{thread-local state}, which may model local
databases, permission matrices, session information, etc. If the term
$t$ in $\action{assert}(t)$ evaluates to a non-true value, then its
containing thread gets stuck forever, else $\action{assert}(t)$ has no
effect.

We abbreviate $(b: x = \alpha); x$ to $b: \alpha$ and $(b: x
= \alpha); e$ to $(b: \alpha); e$ when $x$ is not free in $e$. As an
example, the following expression receives a message, generates a
nonce (through a primitive function $\action{new}$) and sends the
concatenation of the received message and the nonce on the network to
the intended recipient $j$ (line numbers are omitted here).
\[
\begin{array}{l}
m = \action{recv}();  ~~// \text{receive message, bind to $m$}\\
n = \action{new}(); ~~// \text{generate nonce, bind to $n$}\\
\action{send}(j, (m,n));  ~~// \text{send $(m,n)$ to $j$}\\ 
\end{array}
\]
For the purpose of this paper, we limit attention to this simple
expression language, without recursion or branching. Our definition of
actual cause is general and applies to any formalism of
(non-deterministic) interacting agents, but the auxiliary definitions
of log projection and the function $\sf dummify$ introduced later
must be modified. 

\paragraph{Operational Semantics}
\label{app:operational:semantics}
The language $L$'s operational semantics define how a collection
of \emph{threads} execute concurrently. Each thread $T$ contains a
unique thread identifier $i$ (drawn from a universal set of such
identifiers), the executing expression $e$, and a local store.
A \emph{configuration} ${\cal C} = T_1,\ldots,T_n$ models the threads
$T_1,\ldots,T_n$ executing concurrently. Our reduction relation is
written ${\cal C} \rightarrow {\cal C'}$ and defined in the standard
way by interleaving small steps of individual threads (the reduction
relation is parametrized by a semantics of primitive functions
$\zeta$). Importantly, each reduction can either be internal to a
single thread or a \emph{synchronization} of a \action{send} in one
thread with a \action{recv} in another thread.

We make the locus of a reduction explicit by annotating the reduction
arrow with a \emph{label} $\redlabel$. This is written ${\cal
C} \xrightarrow{\redlabel} {\cal C'}$. A label is either the
identifier of a thread $i$ paired with a line number $b$, written
$\langle i, b \rangle$ and representing an internal reduction of some
$\zeta(t)$ in thread $i$ at line number $b$, or a tuple
$\langle \langle i_s, b_s \rangle, \langle i_r, b_r \rangle \rangle$,
representing a synchronization between a \action{send} at line number
$b_s$ in thread $i_s$ with a \action{recv} at line number $b_r$ in
thread $i_r$, or $\epsilon$ indicating an unobservable reduction (of
$t$ or $\action{assert}(t)$) in some thread. Labels $\langle i,
b \rangle$ are called \emph{local labels}, labels $\langle \langle
i_s, b_s \rangle, \langle i_r, b_r \rangle \rangle$ are
called \emph{synchronization labels} and labels $\epsilon$ are
called \emph{silent labels}.

An \emph{initial configuration} can be described by a triple $\langle
I, {\cal A}, \Sigma \rangle$, where $I$ is a finite set of thread
identifiers, ${\cal A}: I \rightarrow \mbox{Expressions}$ and $\Sigma:
I \rightarrow \mbox{Stores}$. This defines an initial configuration of
$|I|$ threads with identifiers in $I$, where thread $i$ contains the
expression ${\cal A}(i)$ and the store $\Sigma(i)$. In the sequel, we
identify the triple $\langle I, {\cal A}, \Sigma \rangle$ with the
configuration defined by it. We also use a configuration's identifiers
to refer to its threads.

\begin{definition}[Run]
Given an initial configuration ${\cal C}_0 = \langle I, {\cal
A}, \Sigma \rangle$, a run is a finite sequence of labeled reductions
${\cal C}_0 \xrightarrow{\redlabel_1} {\cal
C}_1 \ldots \xrightarrow{\redlabel_{n}} {\cal C}_n$.
\end{definition}

A pre-trace is obtained by projecting only the stores from each
configuration in a run.

\begin{definition}[Pre-trace]
Let ${\cal C}_0 \xrightarrow{\redlabel_1} {\cal
C}_1 \ldots \xrightarrow{\redlabel_{n}} {\cal C}_n$ be a run and let
$\Sigma_i$ be the store in configuration ${\cal C}_i$. Then, the
pre-trace of the run is the sequence $(\_, \Sigma_0),
(\redlabel_1,\Sigma_1), \ldots, (\redlabel_n, \Sigma_n)$.
\end{definition}

If $r_i = \epsilon$, then the $i$th step is an unobservable reduction
in some thread and, additionally, $\Sigma_{i-1} = \Sigma_i$. A trace
is a pre-trace from which such $\epsilon$ steps have been dropped.

\begin{definition}[Trace]
The trace of the pre-trace $(\_, \Sigma_0),
(\redlabel_1,\Sigma_1), \ldots, (\redlabel_n, \Sigma_n)$ is the
subsequence obtained by dropping all tuples of the form
$(\epsilon, \Sigma_i)$. Traces are denoted with the letter $t$.
\end{definition}



\subsection{Logs and their projections} \label{sec:support-defs}
To define actual causation, we find it convenient to introduce the
notion of a log and the log of a trace, which is just the sequence of
non-silent labels on the trace. A log is a sequence of labels other
than $\epsilon$. The letter $l$ denotes logs. 

\begin{definition}[Log]
Given a trace $t = (\_,\Sigma_0), (r_1, \Sigma_1), \ldots,
(r_n, \Sigma_n)$, the log of the trace, $log(t)$, is the sequence of
$r_1,\ldots,r_{m}$. (The trace $t$ does not contain a label $r_i$ that
equals $\epsilon$, so neither does $log(t)$.)
\end{definition}

We need a few more straightforward
definitions on logs in order to define actual causation.

\begin{definition}[Projection of a log]
Given a log $l$ and a thread identifier $i$, the projection of $l$ to
$i$, written $l|_i$ is the subsequence of all labels in $l$
that mention $i$. Formally,
\[\begin{array}{l@{~}l@{~}ll}
\bullet|_i & = & \bullet \\
(\langle i, b \rangle :: l)|_i & = & \langle i, b \rangle :: (l|_i) \\
(\langle i, b \rangle :: l)|_i & = & l|_i ~~~~~~ \mbox{if } i\not= j\\
(\langle \langle i_s, b_s \rangle, \langle i_r, b_r \rangle \rangle ::
l)|_i & = & \langle \langle i_s, b_s \rangle, \langle i_r,
b_r \rangle \rangle :: (l|_i) \\
& &  \mbox{if } i_s = i \mbox{ or } i_r = i\\
(\langle \langle i_s, b_s \rangle, \langle i_r, b_r \rangle \rangle ::
l)|_i & = & l|_i \\ 
& & \mbox{if } i_s \not= i \mbox{ and } i_r \not= i
\end{array}\]
\end{definition}

\begin{definition}[Projected prefix]
We call a log $l'$ a
\emph{projected prefix} of the log $l$, written $l' \leq_p l$, if for every
thread identifier $i$, the sequence $l'|_i$ is a prefix of the
sequence $l|_i$.
\end{definition}

The definition of projected prefix allows the relative order of events
in two different non-communicating threads to differ in $l$ and $l'$
but Lamport's happens-before order of actions~\cite{lamport1978} in $l'$ must be
preserved in $l$. Similar to projected prefix, we define projected
sublog.

\begin{definition}[Projected sublog]
We call a log $l'$ a \emph{projected sublog} of the log $l$, written
$l' \sqsubseteq_p l$, if for every thread identifier $i$, the sequence
$l'|_i$ is a subsequence of the sequence $l|_i$ (i.e., dropping some
labels from $l|_i$ results in $l'|_i$).
\end{definition}



\subsection{Properties of Interest}
\label{sec:properties}

A \emph{property} is a set of (good) traces and violations are traces
in the complement of the set.  Our goal is to define the cause of a
violation of a property. We are specifically interested in ascribing
causes to violations of safety
properties~\cite{lamport77:proving} because safety
properties encompass many relevant security requirements. We
recapitulate the definition of a safety property below. Briefly, a
property is safety if it is fully characterized by a set of finite
violating prefixes of traces. Let $U$ denote the universe of all
possible traces.

\begin{definition}[Safety property~\cite{alpern85:liveness}]
A property $P$ (a set of traces) is a safety property, written ${\sf
Safety}(P)$, if $\forall t \not \in P.~\exists t' \in U.~(t' \mbox{ is
a prefix of } t) \cconj (\forall t'' \in U. ~(t' \cdot t'' \not \in
P))$.
\end{definition}

As we explain soon, our causal analysis ascribes thread actions (or
threads) as causes. One important requirement for such analysis is
that the property be closed under reordering of actions in different
threads if those actions are not related by Lamport's happens-before
relation~\cite{lamport1978}. For properties that are not closed in this sense,
the \emph{global order} between actions in a race condition may be a
cause of a violation. Whereas causal analysis of race conditions may
be practically relevant in some situation, we limit attention only to
properties that are closed in the sense described here. We call such
properties reordering-closed or ${\sf RC}$.

\begin{definition}[Reordering-equivalence]
Two traces $t_1, t_2$ starting from the same initial configuration are
called reordering-equivalent, written $t_1 \sim t_2$ if for each
thread identifier $i$, $log(t_1)|_i = log(t_2)|_i$. Note that $\sim$
is an equivalence relation on traces from a given initial
configuration. Let $[t]_\sim$ denote the equivalence class of $t$.
\end{definition}

\begin{definition}[Reordering-closed property]
A property $P$ is called reordering-closed, written ${\sf RC}(P)$, if
$t \in P$ implies $[t]_\sim \subseteq P$. Note that ${\sf RC}(P)$ iff
${\sf RC}(\neg P)$.
\end{definition}

\subsection{Program Actions as Actual Causes} \label{sec:cause1}

In the sequel, let $\varphi_V$ denote the \emph{complement} of a
reordering-closed safety property of interest. (The subscript $V$
stands for ``violations''.) Consider a trace $t$ starting from the
initial configuration ${\cal C}_0 = \langle I, {\cal
A}, \Sigma \rangle$. If $t \in \varphi_V$, then $t$ violates the
property $\neg \varphi_V$.

\begin{definition}[Violation]
A violation of the property $\neg \varphi_V$ is a trace
$t \in \varphi_V$.
\end{definition}

Our definition of actual causation identifies a subset of actions in
$\{{\cal A}(i) ~|~ i \in I\}$ as the cause of a violation
$t \in \varphi_V$. The definition applies in two phases. The first
phase identifies what we call \emph{Lamport causes}. A Lamport cause
is a minimal projected prefix of the log of a violating trace that can
account for the violation. In the second phase, we refine the log by
removing actions that are merely \emph{progress enablers}; the
remaining actions on the log are the \emph{actual action causes}. The
former contribute only indirectly to the cause by enabling the actual
action causes to make progress; the exact values returned by progress
enabling actions are irrelevant.

The following definition, called Phase 1, determines Lamport
causes. It works as follows. We first identify a projected prefix $l$
of the log of a violating trace $t$ as a potential candidate for a
Lamport cause. We then check two conditions on
$l$. The \emph{sufficiency} condition tests that the threads of the
configuration, when executed at least up to the identified prefix,
preserving all synchronizations in the prefix, suffice to recreate the
violation. The \emph{minimality} condition tests that the identified
Lamport cause contains no redundant actions.

\begin{definition}[Phase 1: Lamport Cause of Violation] \label{definition:cause1'} 
Let $t \in \varphi_V$ be a trace starting from ${\cal C}_0 = \langle
I, {\cal A}, \Sigma \rangle$ and $l$ be a projected prefix of
$log(t)$, i.e., $l \leq_p log(t)$. We say that $l$ is the Lamport
cause of the violation $t$ of $\varphi_V$ if the following hold:
\begin{enumerate} 
\item \label{sufficiency1} \textbf{(Sufficiency)} 
Let $T$ be the set of traces starting from ${\cal C}_0$ whose logs
contain $l$ as a projected prefix, i.e., $T = \{t' ~|~ t' \mbox{ is a
trace starting from } {\cal C}_0 \mbox{ and } l \leq_p
log(t')\}$. Then, every trace in $T$ has the violation $\varphi_V$,
i.e., $T \subseteq \varphi_V$. (Because $t \in T$, $T$ is non-empty.)
\item \label{minimality1} \textbf{(Minimality)} 
No proper prefix of $l$ satisfies condition~\ref{sufficiency1}.
\end{enumerate} 
\end{definition}

At the end of Phase 1, we obtain one or more minimal prefixes $l$
which contain program actions that are sufficient for the
violation. These prefixes represent independent Lamport causes of the
violation. In the Phase 2 definition below, we further identify a
sublog $a_d$ of each $l$, such that the program actions in $a_d$ are
actual causes and the actions in $l \backslash a_d$ are progress
enabling actions which only contribute towards the \emph{progress} of
actions in $a_d$ that cause the violation. In other words, the actions
not considered in $a_d$ contain all labels whose actual returned
values are irrelevant.

Briefly, here's how our Phase 2 definition works. We first pick a
candidate projected sublog $a_d$ of $l$, where log $l$ is a Lamport
cause identified in Phase 1. We consider counterfactual traces
obtained from initial configurations in which program actions omitted
from $a_d$ are replaced by actions that do not have any effect other
than enabling the program to progress (referred to as no-op).  If a
violation appears in all such counterfactual traces, then this sublog
$a_d$ is a good candidate. Of all such good candidates, we choose
those that are minimal.

The key technical difficulty in writing this definition is replacing
program actions omitted from $a_d$ with no-ops. We cannot simply erase
any such action because the action is expected to return a term which
is bound to a variable used in the action's continuation. Hence, our
approach is to substitute the variables binding the returns of
no-op'ed actions with arbitrary (side-effect free) terms
$t$. Formally, we assume a function $f:
I \times \mbox{LineNumbers} \rightarrow \mbox{Terms}$ that for line
number $b$ in thread $i$ suggests a suitable term $f(i,b)$ that must
be returned if the action from line $b$ in thread $i$ is replaced with
a no-op. In our cause definition we universally quantify over $f$,
thus obtaining the effect of a no-op. For technical convenience, we
define a syntactic transform called $\sf{dummify}()$ that takes an
initial configuration, the chosen sublog $a_d$ and the function $f$,
and produces a new initial configuration obtained by erasing actions
not in $a_d$ by terms obtained through $f$.

\begin{definition}[Dummifying transformation]\label{def:dummify}
Let $\langle I, {\cal A}, \Sigma \rangle$ be a configuration and let
$a_d$ be a log.  Let $f:
I \times \mbox{LineNumbers} \rightarrow \mbox{Terms}$. The dummifying
transform ${\sf dummify}(I, {\cal A}, \Sigma, a_d, f)$ is the initial
configuration $\langle I, {\cal D}, \Sigma\rangle$, where for all
$i \in I$, ${\cal D}(i)$ is ${\cal A}(i)$ modified as follows:
\begin{itemize} 
\item 
If $(b: x = \action{send}(t)); e$ appears in ${\cal A}(i)$ but
    $\langle i, b \rangle$ does not appear in $a_d$, then replace $(b:
    x = \action{send}(t)); e$ with $e[0/x]$ in ${\cal A}(i)$.
\item 
If $(b: x = \alpha); e$ appears in ${\cal A}(i)$ but $\langle i,
    b \rangle$ does not appear in $a_d$ and
    $\alpha \not= \action{send}(\_)$, then replace $(b: x = \alpha);
    e$ with $e[f(i,b)/x]$ in ${\cal A}(i)$.
\end{itemize}
\end{definition}

We now present our main definition of actual causes.

\begin{definition}[Phase 2: Actual Cause of Violation] 
\label{definition:cause2a} 
Let $t \in \varphi_V$ be a trace from the initial configuration
$\langle I, {\cal A}, \Sigma \rangle$ and let the log $l \leq_p log(t)$
be a Lamport cause of the violation determined by
Definition~\ref{definition:cause1'}. Let $a_d$ be a projected sublog
of $l$, i.e., let $a_d \sqsubseteq_p l$. We say that $a_d$ is the
actual cause of violation $t$ of $\varphi_V$ if the following hold:
\begin{enumerate}
\item \label{sufficiency2} (\textbf{Sufficiency'})
Pick any $f$. Let ${\cal C}_0' = {\sf dummify}(I, {\cal A}, \Sigma,
a_d, f)$ and let $T$ be the set of traces starting from ${\cal C}_0'$
whose logs contain $a_d$ as a projected sublog, i.e., $T = \{t' ~|~
t' \mbox{ is a trace starting from } {\cal C}_0' \mbox{ and }
a_d \sqsubseteq_p log(t')\}$. Then, $T$ is non-empty and every trace
in $T$ has the violation $\varphi_V$, i.e, $T \subseteq \varphi_V$.
\item \label{minimality2} (\textbf{Minimality'})
No proper sublog of $a_d$ satisfies condition~\ref{sufficiency2}.
\end{enumerate}
\end{definition}

At the end of Phase~2, we obtain one or more sets of actions
$a_d$. These sets are deemed the independent actual causes of the
violation $t$.

The following theorem states that for all safety properties that are
re-ordering closed, the Phase 1 and Phase 2 definitions always
identify at least one Lamport and at least one actual cause.

\begin{theorem}
Suppose $\varphi_V$ is reordering-closed and the complement of a
safety property, i.e., ${\sf RC}(\varphi_V)$ and ${\sf
safety}(\neg \varphi_V)$. Then, for every $t \in \varphi_V$: (1)~Our
Phase 1 definition (Definition~\ref{definition:cause1'}) finds a
Lamport cause $l$, and (2)~For every such Lamport cause $l$, the Phase
2 definition (Definition~\ref{definition:cause2a}) finds an actual
cause $a_d$.
\end{theorem}
\begin{proof} 
(1)~Pick any $t \in \varphi_V$. We follow the Phase 1 definition. It
suffices to prove that there is a log $l \leq_p log(t)$ that satisfies
the sufficiency condition. Since ${\sf safety}(\neg \varphi_V)$,
there is a prefix $t_0$ of $t$ s.t. for all $t_1 \in U$,
$t_0 \cdot t_1 \in \varphi_V$. Choose $l = log(t_0)$. Since $t_0$ is a
prefix of $t$, $l = log(t_0) \leq_p log(t)$.
To prove sufficiency, pick any trace $t'$ s.t. $l \leq_p
log(t')$. It suffices to prove $t' \in \varphi_V$. Since $l \leq_p
log(t')$, for each $i$, $log(t')|_i = l|_i \cdot l'_i$ for some
$l'_i$. Let $t''$ be the (unique) subsequence of $t'$ containing all
labels from the logs $\{l_i'\}$. Consider the trace $s = t_0 \cdot
t''$. First, $s$ extends $t_0$, so $s \in \varphi_V$. Second, $s \sim
t'$ because $log(s)|_i = l|_i \cdot l'_i = log(t_0)|_i \cdot
log(t'')|_i = log(t_0 \cdot t'')|_i = log(t')|_i$. Since ${\sf
RC}(\varphi_V)$, $t' \in \varphi_V$. 

(2)~Pick any $t \in \varphi_V$ and let $l$ be a Lamport cause of $t$
as determined by the Phase 1 definition. Following the Phase~2
definition, we only need to prove that there is at least one
$a_d \sqsubseteq_p l$ that satisfies the sufficiency' condition. We
choose $a_d = l$. To show sufficiency', pick any $f$. Because $a_d =
l$, $a_d$ specifies an initial prefix of every ${\cal A}(i)$ and the
transform ${\sf dummify}()$ has no effect on this prefix.  First, we
need to show that at least one trace $t'$ starting from ${\sf
dummify}(I,{\cal A},\Sigma,a_d,f)$ satisfies $a_d \sqsubseteq_p
log(t')$. For this, we can pick $t' = t$. Second, we need to prove
that any trace $t'$ starting from ${\sf dummify}(I,{\cal
A},\Sigma,a_d,f)$ s.t. $a_d \sqsubseteq_p log(t')$ satisfies
$t' \in \varphi_V$. Pick such a $t'$. Let $t_0$ be the prefix of $t$
corresponding to $l$. Then, $log(t_0)|_i = l|_i$ for each $i$. It
follows immediately that for each $i$, $t'|_i = t_0|_i \cdot t''_i$
for some $t''_i$. Let $t''$ be the unique subsequence of $t'$
containing all labels from traces $\{t''_i\}$. Let $s = t_0 \cdot
t''$. First, because for each $i$, $l|_i = log(t_0)|_i$, $l \leq_p
log(t_0)$ trivially.  Because $l$ is a Lamport cause, it satisfies the
sufficiency condition of Phase 1, so $t_0 \in \varphi_V$. Since ${\sf
safety}(\neg \varphi_V)$, and $s$ extends $t_0$,
$s \in \varphi_V$. Second, $s \sim t'$ because $log(s)|_i =
log(t_0)|_i \cdot log(t'')|_i = log(t')|_i$ and both $s$ and $t'$ are
traces starting from the initial configuration ${\sf dummify}(I,{\cal
A},\Sigma,a_d,f)$. Hence, by ${\sf RC}(\varphi_V)$,
$t' \in \varphi_V$.
\end{proof}

Our Phase 2 definition identifies a set of program actions as causes
of a violation. However, in some applications it may be necessary to
ascribe thread identifiers (or programs) as causes. This can be
straightforwardly handled by lifting the Phase 2 definition: A thread
$i$ (or ${\cal A}(i)$) is a cause if one of its actions appears in
$a_d$.

\begin{definition}[Program Cause of Violation] \label{definition:cause2b} 
Let $a_d$ be an actual cause of violation $\varphi_V$ on trace $t$
starting from $\langle I, {\cal A}, \Sigma\rangle$. We say that the
set $X \subseteq I$ of thread identifiers is a cause of the violation
if $X = \{i ~|~ i \mbox{ appears in } a_d\}$.
\end{definition}

\paragraph{Remarks}

We make a few technical observations about our definitions of
cause. First, because Lamport causes
(Definition~\ref{definition:cause1'}) are projected \emph{prefixes},
they contain all actions that occur before any action that actually
contributes to the violation. Many of actions in the Lamport cause may
not contribute to the violation intuitively. Our actual cause
definition filters out such ``spurious'' actions. As an example,
suppose that a safety property requires that the value $1$ never be
sent on the network. The (only) trace of the program $x = 1; y = 2;
z=3; \action{send}(x)$ violates this property. The Lamport cause of
this violation contains all four actions of the program, but it is
intuitively clear that the two actions $y = 2$ and $z = 3$ do not
contribute to the violation. Indeed, the actual cause of the violation
determined by Definition~\ref{definition:cause2a} does not contain
these two actions; it contains only $x = 1$ and $\action{send}(x)$,
both of which obviously contribute to the violation.

Second, our definition of dummification is based on a program
transformation that needs line numbers. One possibly unwanted
consequence is that our traces have line numbers and, hence, we could,
in principle, specify safety properties that are sensitive to line
numbers. However, our definitions of cause are closed under bijective
renaming of line numbers, so if a safety property is insensitive to
line numbers, the actual causes can be quotiented under bijective
renamings of line numbers.

Third, our definition of actual cause
(Definition~\ref{definition:cause2a}) separates actions whose return
values are relevant to the violation from those whose return values
are irrelevant for the violation. This is closely related to
noninterference-like security definitions for information flow
control, in particular, those that separate input presence from input
content~\cite{rafnsson2012}. Lamport causes
(Definition~\ref{definition:cause1'}) have a trivial connection to
information flow: If an action does not occur in any Lamport cause of
a violation, then there cannot be an information flow from that action
to the occurrence of the violation.


\section{Causes of Authentication Failures} \label{sec:application} 

In this section, we model an instance of our running example based on
passwords (Section~\ref{sec:example}) in order to demonstrate our actual cause definition. As
explained in Section~\ref{sec:example}, we consider a protocol session
where \Server1, \User1, \User2, \User3\ and multiple notaries interact over an
adversarial network to establish access over a password-protected
account. We describe a formal model of the protocol
in our language, examine the attack scenario from
Section~\ref{sec:example} and provide a cause analysis using the
definitions from Section~\ref{sec:definitions}.

\subsection{Protocol Description}
\label{subsec:protocol}
We consider our example protocol with eight threads
named \{\Server1, \User1, \Adversary, \Notary1, \Notary2, \Notary3, \User2, \User3\}. In
this section, we briefly describe the protocol and the programs
specified by the protocol for each of these threads. For this purpose,
we assume that we are provided a function ${\cal N}:
I \rightarrow \mbox{Expressions}$ such that ${\cal N}(i)$ is the
program that \emph{ideally should have been} executing in the thread
$i$. For each $i$, we call ${\cal N}(i)$ the \emph{norm} for thread $i$. The
violation is caused because some of the executing programs are
different from the norms. These actual programs, called ${\cal A}$ as
in Section~\ref{sec:definitions}, are shown later. The norms are shown
here to help the reader understand what the ideal protocol is and also
to facilitate some of the development in
Section~\ref{sec:domains}. The appendix describes an
expansion of this example with more than the eight threads considered
here to illustrate our definitions better. The
proof included in the appendix deals with timestamps and
signatures.

The norms in Figure~\ref{fig:norms1'} and the actuals in Figure~\ref{fig:actuals1'} assume that {\User1}'s account
(called $acct$ in {\Server1}'s program) has already been created and
that \User1's password, $pwd$ is associated with {\User1}'s user id,
$uid$. This association (in hashed form) is stored in {\Server1}'s
local state at pointer $mem$. The norm for {\Server1} is to wait for a
request from an entity, respond with its (\Server1's) public key, wait
for a username-password pair encrypted with that public key and grant
access to the requester if the password matches the previously stored
value in \Server1's memory at $mem$. To grant access, \Server1\ adds an
entry into a private access matrix, called $P$. (A separate server
thread, not shown here, allows {\User1} to access its account if this
entry exists in~$P$.)

The norm for {\User1} is to send an access request to {\Server1}, wait
for the server's public key, verify that key with three notaries and
then send its password $pwd$ to {\Server1}, encrypted under \Server1's
public key. On receiving \Server1's public key, {\User1} initiates a
protocol with the three notaries and accepts or rejects the key based
on the response of a majority of the notaries. For simplicity, we omit
a detailed description of this protocol between {\User1} and the
notaries that authenticates the notaries and ensures freshness of
their responses. These details are included in our appendix. In
parallel, the norm for \User2\ is to generate and send a nonce to \User3.  The norm for \User3 is to receive a message from \User2,
generate a nonce and send it to \User2. 

Each notary has a private database of \textit{(public\_key,
principal)} tuples. The notaries' norms assume that this database has
already been created correctly. When {\User1} sends a request with a
public key, the notary responds with the principal's identifier after
retrieving the tuple corresponding to the key from its database.

\paragraph{Notation}  
The programs in this example use several primitive functions
$\zeta$. $\pred{Enc}(k,m)$ and $\pred{Dec}(k',m)$ denote encryption
and decryption of message $m$ with key $k$ and $k'$
respectively. $\pred{Hash}(m)$ generates the hash of term
$m$. $\pred{Sig}(k,m)$ denotes message $m$ signed with the key $k$,
paired with $m$ in the clear. $pub\_key\_i$ and $pvt\_key\_i$ denote
the public and private keys of thread $i$, respectively. For
readability, we include the intended recipient $i$ and expected sender
$j$ of a message as the first argument of $\action{send}(i,m)$ and
$\action{recv}(j)$ expressions. As explained earlier, $i$ and $j$ are
ignored during execution and a network adversary, if present, may
capture or inject any messages.

\paragraph{Security property}

The security property of interest to us is that if at time $u$, a
thread $k$ is given access to account $a$, then $k$ owns
$a$. Specifically, in this example, we are interested in case $a =
acct$ and $k = \User1$. This can be formalized by the following
logical formula, $\neg \varphi_V$:
\vspace{-0.5mm}
\begin{equation}\label{eqn:property}
 \forall u,k. ~(acct, k) \in
P(u) \imp (k = \User1)
\end{equation}
Here, $P(u)$ is the state of the access control matrix $P$ for
\Server1 at time $u$.

\begin{figure}[t]
\small
\begin{tabular}{p{0.5\textwidth}}
\framebox{\noindent\textbf{Norm ${\cal N}(\Server1)$:}}\\
\begin{tabular}{l}
$
\begin{array}{l}
1: \_ = \action{recv} (j);     ~~// \text{access req from thread $j$}\\
2: \action{send}(j, pub\_key\_\Server1);      ~~// \text{send public key to $j$}\\ 
3: s = \action{recv}(j);          ~~// \text{encrypted $uid,pwd$, thread id $J$}\\
4: (uid, pwd, J) = \pred{Dec}(pvt\_key\_\Server1, s); \\
5: t = \action{hash}(uid, pwd); \\
\mbox{  }~~\action{assert} (mem = t)  ~~// \text{compare hash with stored value}\\
6:  \action{insert}(P, (acct,J)); \\
\end{array} $\\
\end{tabular}

\framebox{\noindent\textbf{Norm ${\cal N}(\User1)$:}}\\
\begin{tabular}{l}
$
\begin{array}{l}
1: \action{send}(\Server1); ~~// \text{access request}\\
2: pub\_key = \action{recv}(\Server1); ~~//\text{key from $\Server1$}\\
3: \action{send}(\Notary1, pub\_key);\\
4: \action{send}(\Notary2, pub\_key);\\
5: \action{send}(\Notary3, pub\_key);\\
6: \pred{Sig}(pub\_key, l1) = \action{recv}(\Notary1); ~~//\text{notary1 responds }\\
7: \pred{Sig}(pub\_key, l2) = \action{recv}(\Notary2); ~~//\text{notary2 responds}\\
8: \pred{Sig}(pub\_key, l3) = \action{recv}(\Notary3); ~~//\text{notary3 responds}\\

 \mbox{  }~~\action{assert} (\text{At least two of \{l1,l2,l3\} equal \Server1}) \\
9: t = \pred{Enc}(pub\_key, (uid, pwd, \User1)); \\
10:\action{send}(\Server1, t); ~~// \text{send $t$ to \Server1}\\
\end{array} $
\end{tabular}\\

\framebox{\noindent\textbf{Norms ${\cal N}(\Notary1), {\cal N}(\Notary2), {\cal N}(\Notary3)$:}}\\
\begin{tabular}{l}
$
\begin{array}{l}
//\text{ $o$ denotes \Notary1, \Notary2 or \Notary3}\\
1:pub\_key = \action{recv}(j); \\
2:pr = \pred{KeyOwner}(pub\_key); ~~~// \text{lookup key owner} \\%
3:\action{send}(j, \pred{Sig}(pvt\_key\_o, (pub\_key, pr)); \\%
\end{array} $
\end{tabular}\\

\framebox{\noindent\textbf{Norm ${\cal N}(\User2)$:}}\\		\begin{tabular}{l}
		$
		\begin{array}{l}
		1: \action{send}(\User3); \\
		2: \_ = \action{recv}(\User3);
		\end{array} $ 
                \end{tabular}
		\\

	\framebox{\noindent\textbf{Norm ${\cal N}(\User3)$:}}\\		
	\begin{tabular}{l}
		$
		\begin{array}{l}
		1:  \_ = \action{recv}(\User2);\\
		2:  \action{send}(\User3); \\
		\end{array} $ 
                \end{tabular}

\end{tabular}

\caption{Norms for all threads. \Adversary's norm is the trivial empty program.}
\label{fig:norms1'}
\end{figure}

\begin{figure}[t]
\small
\begin{tabular}{p{0.5\textwidth}}
\framebox{\noindent\textbf{Actual ${\cal A}(\Server1)$:}}\\
\begin{tabular}{l}
$
\begin{array}{l}
1: \_ = \action{recv} (j);     ~~// \text{access req from thread $j$}\\
2: \action{send}(j, pub\_key\_\Server1);      ~~// \text{send public key to $j$}\\ 
3: \_ = \action{recv}(j);  ~~// \text{receive nonce from thread \User2}\\
4: \action{send}(j);~~//\text{send signed nonce}\\ 
5: s = \action{recv}(j);          ~~// \text{encrypted $uid,pwd$, thread id from $j$}\\
6: (uid, pwd, J) = \pred{Dec}(pvt\_key\_\Server1, s); \\
7: t = \action{hash}(uid, pwd); \\
\mbox{  }~~\action{assert} (mem = t)  \textbf{[A]}~~// \text{compare hash with stored value}\\
8:  \action{insert}(P, (acct,J)); \\
 
\end{array} $\\
\end{tabular}

\framebox{\noindent\textbf{Actual ${\cal A}(\User1)$:}}\\
\begin{tabular}{l}
$
\begin{array}{l}
1: \action{send}(\Server1); ~~// \text{access request}\\
2: pub\_key = \action{recv}(\Server1); ~~//\text{key from $\Server1$}\\

3: \action{send}(\Notary1, pub\_key);\\
4: \action{send}(\Notary2, pub\_key);\\
5: \action{send}(\Notary3, pub\_key);\\
6: \pred{Sig}(pub\_key, l1) = \action{recv}(\Notary1); ~~//\text{notary1 responds }\\
7: \pred{Sig}(pub\_key, l2) = \action{recv}(\Notary2); ~~//\text{notary2 responds}\\
8: \pred{Sig}(pub\_key, l3) = \action{recv}(\Notary3); ~~//\text{notary3 responds}\\

\mbox{  }~~\action{assert} (\text{At least two of \{l1,l2,l3\} equal \Server1}) \textbf{[B]}\\
9: t = \pred{Enc}(pub\_key, (uid, pwd, \User1)); \\%
10: \action{send}(\Server1, t); ~~// \text{send $t$ to \Server1}\\
\end{array} $
\end{tabular}\\

\framebox{\noindent\textbf{Actual ${\cal A}(\Adversary)$ }}\

		\begin{tabular}{l}
		$
		\begin{array}{l}
		1: \action{recv} (\User1);     ~~// \text{intercept access req from \User1}\\
		2: \action{send}(\User1, pub\_key\_A); ~~// \text{send key to User}\\ 
		3: s = \action{recv}(\User1); ~~// \text{pwd from \User1}\\
		4: (uid, pwd, \User1) = \pred{Dec}(pvt\_key\_A, s); ~~// \text{decrypt pwd}\\
		5: \action{send}(\Server1, uid); ~~// \text{access request to \Server1}\\
		6: pub\_key = \action{recv}(\Server1); ~~// \text{Receive \Server1's public key} \\
		7: t = \pred{Enc}(pub\_key, (uid, pwd, \Adversary)); // \text{encrypt pwd}\\
		8: \action{send}(\Server1, t); ~~// \text{pwd to \Server1} \\
		\end{array} $ \\
		\end{tabular}

		\framebox{\noindent\textbf{Actuals ${\cal A}(\Notary1), {\cal A}(\Notary2), {\cal N}(\Notary3)$:}}\\

		\begin{tabular}{l}
		$
\begin{array}{l}
//\text{ $o$ denotes \Notary1, \Notary2 or \Notary3}\\

1:pub\_key = \action{recv}(j); \\
2:\action{send}(j, \pred{Sig}(pvt\_key\_o, (pub\_key, \Server1))); \\

\end{array} $
		\end{tabular}\\
		
		\framebox{\noindent\textbf{Actual ${\cal A}(\User2)$:}}\\		
	\begin{tabular}{l}
		$
		\begin{array}{l}
		1: \action{send}(\Server1); ~~//\text{send nonce to \Server1}\\
		2: \_ = \action{recv}(\Server1);\\
		3: \action{send}(\User3); ~~//\text{forward nonce to \User3} \\
		4: \_ = \action{recv}(\User3);\\
		\end{array} $ 
                \end{tabular}\\

	\framebox{\noindent\textbf{Actual ${\cal A}(\User3)$:}}\\		
	\begin{tabular}{l}
		$
		\begin{array}{l}
		1:  \_ = \action{recv}(\User2);\\
		2:  \action{send}(\User2); ~~//\text{send nonce to \User2}\\
		\end{array} $ 
		\end{tabular}\\

\end{tabular}

\caption{Actuals for all threads.}
\label{fig:actuals1'}
\end{figure}

\subsection{Attack}
As an illustration, we model the ``Compromised Notaries'' violation of Section~\ref{sec:example}.  The programs executed by all threads are given in
Figure~\ref{fig:actuals1'}. \User1\ sends an access request
to \Server1\ which is intercepted by \Adversary\ who sends its own key
to \User1\ (pretending to be \Server1). \User1\ checks with the three
notaries who falsely verify \Adversary's public key to be
\Server1's key. Consequently, \User1\ sends the password to
\Adversary. {\Adversary}\ then initiates a protocol with {\Server1}
and gains access to the \User1's account. Note that the actual
programs of the three notaries attest that the public key given to
them belongs to \Server1. In parallel, \User2\ sends a request
to \Server1 and receives a response from \Server1. Following this
interaction, \User2 interacts with \User3, as in their norms.

Figure~\ref{fig:log1'} shows the expressions executed by each thread
on the property-violating trace. For instance, the label
$\langle \langle \User1, 1 \rangle, \langle \Adversary,
1 \rangle \rangle $ indicates that both \User1 and \Adversary\
executed the expressions with the line number 1 in their actual
programs, which resulted in a synchronous communication between them,
while the label $\langle \Adversary, 4 \rangle $ indicates the local
execution of the expression at line~4 of \Adversary's program. The
initial configuration has the programs: $\{{\cal A}(\User1), {\cal
A}(\Server1), {\cal A}({\Adversary}), {\cal A}({\Notary1}), \\ {\cal
A}({\Notary2}), {\cal A}({\Notary3}), {\cal A}(\User2), {\cal
A}(\User3) \}$. For this attack scenario, the concrete trace $t$ we
consider is such that $\log(t)$ is any
\emph{arbitrary interleaving} of the actions for $X= \{\Adversary, \User1, \User2, \User3, 
\Server1, \Notary1, \\ \Notary2, \Notary3\}$ shown in Figure~\ref{fig:log1'}(a). Any such interleaved log is
denoted $\log(t)$ in the sequel. At the end of this log, $(acct, \Adversary)$ occurs in the access
control matrix $P$, but {\Adversary} does not own $acct$. Hence, this
log corresponds to a violation of our security property.

\begin{figure*}
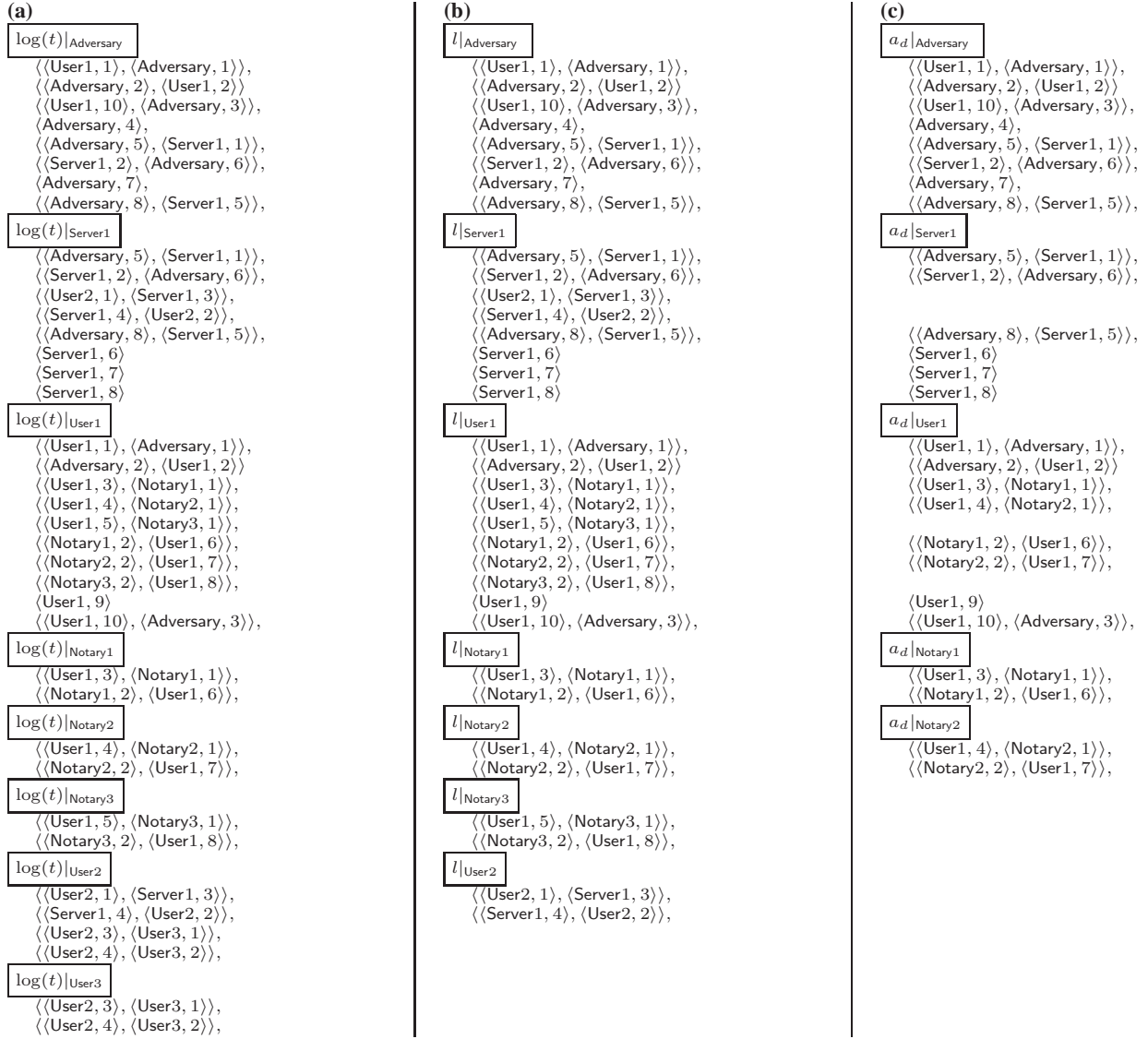


\scriptsize
\begin{tabular}{p{0.33\textwidth}|p{0.33\textwidth}|p{0.34\textwidth}}

\begin{tabular}[t]{l}

\small

	{\noindent\textbf{(a)}}\\
	\framebox{\noindent\textbf{$\log(t)|_{\Adversary}$ }}\\

		\begin{tabular}{l}
		$
		\begin{array}{l}
		
		\langle \langle \User1, 1 \rangle, \langle \Adversary, 1 \rangle \rangle, \\	
			\langle \langle \Adversary, 2  \rangle, \langle \User1, 2 \rangle \rangle \\
			\langle \langle \User1, 10 \rangle, \langle \Adversary, 3 \rangle \rangle,\\
			\langle  \Adversary, 4 \rangle,\\
			\langle \langle \Adversary, 5 \rangle,\langle \Server1, 1 \rangle \rangle,\\
			\langle \langle \Server1, 2 \rangle, \langle \Adversary, 6 \rangle \rangle, \\
			\langle  \Adversary, 7 \rangle,\\
			      		\langle \langle \Adversary, 8  \rangle,  \langle \Server1, 5 \rangle \rangle, \\

		\end{array} $ \\
		\end{tabular}\\

			\framebox{\noindent\textbf{$\log(t)|_{\Server1}$}}\\
		\begin{tabular}{l}
		$
		\begin{array}{l}

		  \langle \langle \Adversary, 5 \rangle,\langle \Server1, 1 \rangle \rangle,\\
				\langle \langle \Server1, 2 \rangle, \langle \Adversary, 6 \rangle \rangle, \\
					\langle \langle \User2, 1 \rangle,  \langle \Server1, 3 \rangle \rangle, \\
					\langle \langle \Server1, 4 \rangle, \langle \User2, 2 \rangle \rangle, \\
				\langle \langle \Adversary, 8 \rangle,  \langle \Server1, 5 \rangle \rangle, \\
               \langle \Server1, 6 \rangle\\ 		  
		       \langle \Server1, 7 \rangle\\ 	
			       \langle \Server1, 8 \rangle\\ 
		 				
		\end{array} $\\
		\end{tabular}\\

		\framebox{\noindent\textbf{$\log(t)|_{\User1}$}}\\
		\begin{tabular}{l}
		$
		\begin{array}{l}
	
		\langle \langle \User1, 1 \rangle, \langle \Adversary, 1 \rangle \rangle, \\
		\langle \langle \Adversary, 2 \rangle, \langle \User1, 2 \rangle \rangle \\
		\langle \langle \User1, 3 \rangle,  \langle \Notary1, 1 \rangle \rangle, \\
		\langle \langle \User1, 4 \rangle,  \langle \Notary2, 1 \rangle \rangle,\\
		\langle \langle \User1, 5 \rangle,  \langle \Notary3, 1 \rangle \rangle,\\

		   \langle \langle \Notary1, 2 \rangle, \langle \User1, 6 \rangle \rangle,  \\
                  \langle \langle \Notary2, 2 \rangle, \langle \User1, 7 \rangle \rangle,  \\
                          \langle \langle \Notary3, 2 \rangle, \langle \User1, 8 \rangle \rangle,  \\
	 		\langle \User1, 9 \rangle \\
			\langle \langle \User1, 10 \rangle, \langle \Adversary, 3 \rangle \rangle,\\
		
			\end{array} $ \\
			\end{tabular}\\

		\framebox{\noindent\textbf{$\log(t)|_{\Notary1}$}}\\

		\begin{tabular}{l}
		$
		\begin{array}{l}
		\langle \langle \User1, 3 \rangle,  \langle \Notary1, 1 \rangle \rangle, \\
                  \langle \langle \Notary1, 2 \rangle, \langle \User1, 6 \rangle \rangle, \\
		\end{array} $ 
		\end{tabular}\\

		\framebox{\noindent\textbf{$\log(t)|_{\Notary2}$}}\\

		\begin{tabular}{l}
		$
		\begin{array}{l}
		\langle \langle \User1, 4 \rangle,  \langle \Notary2, 1 \rangle \rangle, \\
                	\langle \langle \Notary2, 2 \rangle, \langle \User1, 7 \rangle \rangle,  \\
		\end{array} $ 
		\end{tabular}\\

               \framebox{\noindent\textbf{$\log(t)|_{\Notary3}$}}\\

		\begin{tabular}{l}
		$
		\begin{array}{l}
		\langle \langle \User1, 5 \rangle,  \langle \Notary3, 1 \rangle \rangle, \\
                	\langle \langle \Notary3, 2 \rangle, \langle \User1, 8 \rangle \rangle,  \\
		\end{array} $ 
		\end{tabular}\\
   
			\framebox{\noindent\textbf{$\log(t)|_{\User2}$}}\\
			\begin{tabular}{l}
		$
		\begin{array}{l}
	\langle \langle \User2, 1 \rangle,  \langle \Server1, 3 \rangle \rangle, \\
					\langle \langle \Server1, 4 \rangle, \langle \User2, 2 \rangle \rangle, \\
\langle \langle \User2, 3 \rangle,  \langle \User3, 1 \rangle \rangle, \\
\langle \langle \User2, 4 \rangle,  \langle \User3, 2 \rangle \rangle, \\

\end{array} $ 
		\end{tabular}\\
   
		\framebox{\noindent\textbf{$\log(t)|_{\User3}$}}\\
\begin{tabular}{l}
		$
		\begin{array}{l}
		\langle \langle \User2, 3 \rangle,  \langle \User3, 1 \rangle \rangle, \\
\langle \langle \User2, 4 \rangle,  \langle \User3, 2 \rangle \rangle, \\

\end{array} $ 
		\end{tabular}\\

			\end{tabular}
			&
			\begin{tabular}[t]{l}

\small

{\noindent\textbf{(b)}}\\
	\framebox{\noindent\textbf{$l|_{\Adversary}$ }}\\

		\begin{tabular}{l}
		$
		\begin{array}{l}
		
		\langle \langle \User1, 1 \rangle, \langle \Adversary, 1 \rangle \rangle, \\	
			\langle \langle \Adversary, 2  \rangle, \langle \User1, 2 \rangle \rangle \\
			\langle \langle \User1, 10 \rangle, \langle \Adversary, 3 \rangle \rangle,\\
			\langle  \Adversary, 4 \rangle,\\
			\langle \langle \Adversary, 5 \rangle,\langle \Server1, 1 \rangle \rangle,\\
			\langle \langle \Server1, 2 \rangle, \langle \Adversary, 6 \rangle \rangle, \\
			\langle  \Adversary, 7 \rangle,\\
			      		\langle \langle \Adversary, 8  \rangle,  \langle \Server1, 5 \rangle \rangle, \\

		\end{array} $ \\
		\end{tabular}\\

			\framebox{\noindent\textbf{$l|_{\Server1}$}}\\
		\begin{tabular}{l}
		$
		\begin{array}{l}

		  \langle \langle \Adversary, 5 \rangle,\langle \Server1, 1 \rangle \rangle,\\
				\langle \langle \Server1, 2 \rangle, \langle \Adversary, 6 \rangle \rangle, \\
					\langle \langle \User2, 1 \rangle,  \langle \Server1, 3 \rangle \rangle, \\
					\langle \langle \Server1, 4 \rangle, \langle \User2, 2 \rangle \rangle, \\
				\langle \langle \Adversary, 8 \rangle,  \langle \Server1, 5 \rangle \rangle, \\
               \langle \Server1, 6 \rangle\\ 		  
		       \langle \Server1, 7 \rangle\\ 	
			       \langle \Server1, 8 \rangle\\ 
		 				
		\end{array} $\\
		\end{tabular}\\

		\framebox{\noindent\textbf{$l|_{\User1}$}}\\
		\begin{tabular}{l}
		$
		\begin{array}{l}
	
		\langle \langle \User1, 1 \rangle, \langle \Adversary, 1 \rangle \rangle, \\
		\langle \langle \Adversary, 2 \rangle, \langle \User1, 2 \rangle \rangle \\
		\langle \langle \User1, 3 \rangle,  \langle \Notary1, 1 \rangle \rangle, \\
		\langle \langle \User1, 4 \rangle,  \langle \Notary2, 1 \rangle \rangle,\\
		\langle \langle \User1, 5 \rangle,  \langle \Notary3, 1 \rangle \rangle,\\

		   \langle \langle \Notary1, 2 \rangle, \langle \User1, 6 \rangle \rangle,  \\
                  \langle \langle \Notary2, 2 \rangle, \langle \User1, 7 \rangle \rangle,  \\
                          \langle \langle \Notary3, 2 \rangle, \langle \User1, 8 \rangle \rangle,  \\
	 		\langle \User1, 9 \rangle \\
			\langle \langle \User1, 10 \rangle, \langle \Adversary, 3 \rangle \rangle,\\
		
			\end{array} $ \\
			\end{tabular}\\

		\framebox{\noindent\textbf{$l|_{\Notary1}$}}\\

		\begin{tabular}{l}
		$
		\begin{array}{l}
		\langle \langle \User1, 3 \rangle,  \langle \Notary1, 1 \rangle \rangle, \\
                  \langle \langle \Notary1, 2 \rangle, \langle \User1, 6 \rangle \rangle, \\
		\end{array} $ 
		\end{tabular}\\

		\framebox{\noindent\textbf{$l|_{\Notary2}$}}\\

		\begin{tabular}{l}
		$
		\begin{array}{l}
		\langle \langle \User1, 4 \rangle,  \langle \Notary2, 1 \rangle \rangle, \\
                	\langle \langle \Notary2, 2 \rangle, \langle \User1, 7 \rangle \rangle,  \\
		\end{array} $ 
		\end{tabular}\\

               \framebox{\noindent\textbf{$l|_{\Notary3}$}}\\

		\begin{tabular}{l}
		$
		\begin{array}{l}
		\langle \langle \User1, 5 \rangle,  \langle \Notary3, 1 \rangle \rangle, \\
                	\langle \langle \Notary3, 2 \rangle, \langle \User1, 8 \rangle \rangle,  \\
		\end{array} $ 
		\end{tabular}\\
   
			\framebox{\noindent\textbf{$l|_{\User2}$}}\\
			\begin{tabular}{l}
		$
		\begin{array}{l}
	\langle \langle \User2, 1 \rangle,  \langle \Server1, 3 \rangle \rangle, \\
					\langle \langle \Server1, 4 \rangle, \langle \User2, 2 \rangle \rangle, \\

\end{array} $ 
		\end{tabular}\\

			\end{tabular}
			
			&
			
			\begin{tabular}[t]{l}
				
\small
{\noindent\textbf{(c)}}\\
		\framebox{\noindent\textbf{$a_d|_{\Adversary}$ }}\\

		\begin{tabular}{l}
		$
		\begin{array}{l}
		
		\langle \langle \User1, 1 \rangle, \langle \Adversary, 1 \rangle \rangle, \\	
			\langle \langle \Adversary, 2  \rangle, \langle \User1, 2 \rangle \rangle \\
			\langle \langle \User1, 10 \rangle, \langle \Adversary, 3 \rangle \rangle,\\
			\langle  \Adversary, 4 \rangle,\\
			\langle \langle \Adversary, 5 \rangle,\langle \Server1, 1 \rangle \rangle,\\
			\langle \langle \Server1, 2 \rangle, \langle \Adversary, 6 \rangle \rangle, \\
			\langle  \Adversary, 7 \rangle,\\
			      		\langle \langle \Adversary, 8  \rangle,  \langle \Server1, 5 \rangle \rangle, \\

		\end{array} $ \\
		\end{tabular}\\

			\framebox{\noindent\textbf{$a_d|_{\Server1}$}}\\
		\begin{tabular}{l}
		$
		\begin{array}{l}

		  \langle \langle \Adversary, 5 \rangle,\langle \Server1, 1 \rangle \rangle,\\
				\langle \langle \Server1, 2 \rangle, \langle \Adversary, 6 \rangle \rangle, \\
	\\
	\\
				\langle \langle \Adversary, 8 \rangle,  \langle \Server1, 5 \rangle \rangle, \\
               \langle \Server1, 6 \rangle\\ 		  
		       \langle \Server1, 7 \rangle\\ 	
			       \langle \Server1, 8 \rangle\\ 
		 				
		\end{array} $\\
		\end{tabular}\\

		\framebox{\noindent\textbf{$a_d|_{\User1}$}}\\
		\begin{tabular}{l}
		$
		\begin{array}{l}
	
		\langle \langle \User1, 1 \rangle, \langle \Adversary, 1 \rangle \rangle, \\
		\langle \langle \Adversary, 2 \rangle, \langle \User1, 2 \rangle \rangle \\
		\langle \langle \User1, 3 \rangle,  \langle \Notary1, 1 \rangle \rangle, \\
		\langle \langle \User1, 4 \rangle,  \langle \Notary2, 1 \rangle \rangle,\\
	\\
		   \langle \langle \Notary1, 2 \rangle, \langle \User1, 6 \rangle \rangle,  \\
                  \langle \langle \Notary2, 2 \rangle, \langle \User1, 7 \rangle \rangle,  \\
                       \\
	 		\langle \User1, 9 \rangle \\
			\langle \langle \User1, 10 \rangle, \langle \Adversary, 3 \rangle \rangle,\\
		
			\end{array} $ \\
			\end{tabular}\\

		\framebox{\noindent\textbf{$a_d|_{\Notary1}$}}\\

		\begin{tabular}{l}
		$
		\begin{array}{l}
		\langle \langle \User1, 3 \rangle,  \langle \Notary1, 1 \rangle \rangle, \\
                  \langle \langle \Notary1, 2 \rangle, \langle \User1, 6 \rangle \rangle, \\
		\end{array} $ 
		\end{tabular}\\

		\framebox{\noindent\textbf{$a_d|_{\Notary2}$}}\\

		\begin{tabular}{l}
		$
		\begin{array}{l}
		\langle \langle \User1, 4 \rangle,  \langle \Notary2, 1 \rangle \rangle, \\
                	\langle \langle \Notary2, 2 \rangle, \langle \User1, 7 \rangle \rangle,  \\
		\end{array} $ 
		\end{tabular}\\

			\end{tabular}

\end{tabular}
\caption{\emph{Left to Right:} \textbf{(a):} $\log(t)|_i$ for $i \in I$. \textbf{(b):} Lamport cause $l$ for Theorem~\ref{theorem-passwords1}. $l|_i= \emptyset$ for $i \in \{\User3\}$ as output by Definition~\ref{definition:cause1'}. \textbf{(c):} Actual cause $a_d$ for Theorem~\ref{theorem-passwords1}. $a_d|_i= \emptyset$ for $i \in \{\Notary3, \User2, \User3\}$. $a_d$ is a projected \emph{sublog}
of Lamport cause $l$.}
\label{fig:log1'}
\end{figure*}

Note that if any two of the three notaries had attested
the \Adversary's key to belong to \Server1, the violation would still
have happened. Consequently, we may expect three independent program
causes in this
example: \{\Adversary, \User1, \Server1, \Notary1, \Notary2\} with the
action causes $a_d$ as shown in Figure~\ref{fig:log1'}(c),
\{\Adversary, \User1, \Server1, \Notary1, \Notary3\} with the actions $a_d'$, and
\{\Adversary, \User1, \Server1, \Notary2, \Notary3\} with the actions $a_d''$ where $a_d'$ and $a_d''$ can be obtained from $a_d$ (Figure~\ref{fig:log1'}) by considering actions for \{\Notary1, \Notary3\} and \{\Notary2, \Notary3\} respectively, instead of actions for \{\Notary1, \Notary2\}. Our treatment of independent causes follows the tradition in the causality literature. The following
theorem states that our definitions determine exactly these three
independent causes -- one notary is dropped from each of these sets, but no notary is discharged from all the sets. This determination reflects the intuition that only two dishonest notaries are sufficient to cause the violation. Additionally, while it is true that all parties who follow the protocol should not be \emph{blamed} for a violation, an honest party may be an \emph{actual cause} of the violation (in both the common and the philosophical sense of the word), as demonstrated in this case study. This two-tiered view of accountability of an action by separately asserting cause and blame can also be found in prior work in law and philosophy~\cite{nissenbaum1996,feinberg1985suaculpa}. Determining actual cause is nontrivial and is the focus of this work.

\begin{theorem}\label{theorem-passwords1}
Let $I
= \{\User1, \Server1, \Adversary, \Notary1,\\ \Notary2, \Notary3, \User2, \User3\}$
and $\Sigma$ and ${\cal A}$ be as described above. Let $t$ be a trace
from $\langle I, {\cal A}, \Sigma\rangle$ such that $log(t)|_i$ for
each $i \in I$ matches the corresponding log projection from
Figure~\ref{fig:log1'}(a). Then, Definition~\ref{definition:cause2b}
determines three possible values for the program cause $X$ of
violation
$t \in \varphi_V$: \{\Adversary, \User1, \Server1, \Notary1, \Notary2\},
\{\Adversary, \User1, \Server1, \Notary1, \Notary3\}, and
\{\Adversary, \User1, \Server1, \Notary2, \Notary3\} where the 
corresponding actual causes are $a_d, a_d'$ and $a_d''$, respectively.
\end{theorem}


It is instructive to understand the proof of this theorem, as it
illustrates our definitions of causation.  We verify that our Phase~1,
Phase~2 definitions (Definitions~\ref{definition:cause1'},
~\ref{definition:cause2a}, ~\ref{definition:cause2b}) yield exactly
the three values for $X$ mentioned in the theorem.

\paragraph{Lamport cause (Phase 1)}
We show that any $l$ whose projections match those shown in
Figure~\ref{fig:log1'}(b) satisfies sufficiency and minimality. From
Figure~\ref{fig:log1'}(b), such an $l$ has no actions for \User3 and
only those actions of \User2 that are involved in synchronization
with \Server1. For all other threads, the log contains every action
from $t$. The intuitive explanation for this $l$ is straightforward:
Since $l$ must be a (projected) \emph{prefix} of the trace, and the
violation only happens because of $\action{insert}$ in the last
statement of \Server1's program, every action of every program before
that statement in Lamport's happens-before relation must be in
$l$. This is exactly the $l$ described in Figure~\ref{fig:log1'}(b).

Formally, following the statement of sufficiency, let $T$ be the set
of traces starting from ${\cal C}_0 = \langle {I}, {\cal
A}, \Sigma \rangle$ (Figure~\ref{fig:actuals1'}) whose logs contain
$l$ as a projected prefix. Pick any $t' \in T$. We need to show
$t' \in \varphi_V$. However, note that any $t'$ containing all actions
in $l$ must also add $(acct, \Adversary)$ to $P$, but
$\Adversary \not= \User1$. Hence, $t' \in \varphi_V$. Further, $l$ is
minimal as described in the previous paragraph.

\paragraph{Actual cause (Phase 2)}
Phase 2 (Definitions~\ref{definition:cause2a},
~\ref{definition:cause2b}) determines three independent program causes
for
$X$: \{\Adversary, \User1, \Server1, \Notary1, \Notary2\}, \{\Adversary, \User1, \Server1, \Notary1, \Notary3\},
and \{\Adversary, \User1, \Server1, \Notary2, \Notary3\} with the
actual action causes given by $a_d, a_d'$ and $a_d''$, respectively in
Figure~\ref{fig:log1'}. These are symmetric, so we only explain why
$a_d$ satisfies Definition~\ref{definition:cause2a}. (For this $a_d$,
Definition~\ref{definition:cause2b} immediately forces $X
= \{\Adversary, \User1, \Server1, \Notary1, \Notary2\}$.) We show that
(a) $a_d$ satisfies sufficiency', and (b) No proper sublog of $a_d$
satisfies sufficiency' (minimality'). Note that $a_d$ is obtained from
$l$ by dropping \Notary3, \User2 and \User3, and all their
interactions with other threads.

We start with (a). Let $a_d$ be such that $a_{d}|_i$ matches
Figure~\ref{fig:log1'}(c) for every $i$. Fix any dummifying function
$f$. We must show that any trace originating from ${\sf dummify}(I,
{\cal A}, \Sigma,a_d,f)$, whose log contains $a_d$ as a projected
sublog, is in $\varphi_V$. Additionally we must show that there is
such a trace. There are two potential issues in mimicking the
execution in $a_d$ starting from ${\sf dummify}(I, {\cal
A}, \Sigma,a_d,f)$ --- first, with the interaction between \User1
and \Notary3 and, second, with the interaction between \Server1
and \User2. For the first interaction, on line~5, ${\cal A}(\User1)$
(Figure~\ref{fig:actuals1'}) synchronizes with {\Notary3} according to
$l$, but the synchronization label does not exist in $a_d$. However,
in ${\sf dummify}(I, {\cal A}, \Sigma,a_d,f)$, the $\action{recv}()$
on line~8 in ${\cal A}(\User1)$ is replaced with a dummy value, so the
execution from ${\sf dummify}(I, {\cal A}, \Sigma,a_d,f)$
progresses. Subsequently, the majority check (assertion [B]) succeeds
as in $l$, because two of the three notaries ({\Notary1} and
{\Notary2}) still attest the {\Adversary}'s key. A similar observation can be made about the interaction between \Server1 and \User2.

Next we prove that every trace starting from ${\sf dummify}(I, {\cal
A}, \Sigma,a_d,f)$, whose log contains $a_d$
(Figure~\ref{fig:log1'}) as a projected sublog, is in
$\varphi_V$. Fix a trace $t'$ with log $l'$. Assume $l'$ contains
$a_d$. We show $t' \in \varphi_V$ as follows:
\begin{enumerate}
\item 
Since the synchronization labels in $l'$ are a superset of those in
$a_d$, {\Server1} must execute line~8 of its program ${\cal
A}(\Server1)$ in $t'$. After this line, the access control matrix $P$
contains $(acct, J)$ for some $J$.
\item
When ${\cal A}(\Server1)$ writes $(x, J)$ to $P$ at line~8, then $J$
is the third component of a tuple obtained by decrypting a message
received on line~5.
\item
Since the synchronization projections on $l'$ are a superset of $a_d$,
and on $a_d$ $\langle \Server1, 5\rangle$ synchronizes with
$\langle \Adversary, 8\rangle$, $J$ must be the third component of an
encrypted message sent on line~8 of ${\cal A}(\Adversary)$.
\item 
The third component of the message sent on line~8 by {\Adversary} is
exactly the term ``\Adversary''. (This is easy to see, as the term
``\Adversary'' is hardcoded on line 7.) Hence, $J = \Adversary$.
\item 
This immediately implies that $t' \in \varphi_V$ since
$(acct, \Adversary) \in P$, but ${\Adversary} \not= {\User1}$.
\end{enumerate}

Last, we prove (b) --- that no proper subsequence of $a_d$ satisfies
sufficiency'. Note that $a_d$ (Figure~\ref{fig:log1'}(c)) contains
exactly those actions from $l$ (Figure~\ref{fig:log1'}) on whose
returned values the last statement of \Server1's program
(Figure~\ref{fig:actuals1'}) is data or control
dependent. Consequently, all of $a_d$ as shown is necessary to obtain
the violation. 

(The astute reader may note that in Figure~\ref{fig:actuals1'}, there
is no dependency between line~1 of \Server1's program and
the \texttt{insert} statement in \Server1. Hence, line~1 should not be
in $a_d$. While this is accurate, the program in
Figure~\ref{fig:actuals1'} is a slight simplification of the real
protocol, which is shown in the appendix. In the real protocol, line~1
returns a received nonce, whose value does influence whether or not
execution proceeds to the \texttt{insert} statement.)

\section{Towards Accountability}\label{sec:domains}
In this section,  we discuss the use of our causal analysis techniques for providing explanations and assigning blame.

\subsection{Using Causality for Explanations}
Generating explanations involves enhancing the epistemic state of an
agent by providing information about the cause of an
outcome~\cite{halpern2005explanations}. Automating this process is
useful for several tasks such as planning in AI-related applications
and has also been of interest in the philosophy
community~\cite{halpern2005explanations,
woodward2003making}. Causation has also been applied for explaining
counter examples and providing explanations for errors in model
checking~\cite{explaining-counterexamples2009,halpern-specification2008,groce2006error,rajamani2003}
where the abstract nature of the explanation provides insight about
the model.
 
In prior work, Halpern and Pearl have defined explanation in terms of
causality~\cite{halpern2005explanations}. A fact, say $E$,
constitutes an explanation for a previously established fact $F$ in a
given context, if had $E$ been true then it would have been a
sufficient cause of the established fact $F$ . Moreover, having this
information advances the prior epistemic state of the agent seeking
the explanation, i.e. there exists a world (or a setting of the
variables in Halpern and Pearl's model) where $F$ is not true but $E$
is.

Our definition of cause (Section~\ref{sec:definitions}) could be
used to explain violations arising from execution of programs in a
given initial configuration. Given a log $l$, an initial configuration
${\cal C}_0$, and a violation $\varphi_V$, our definition would
pinpoint a sequence of program actions, $a_d$, as an actual
cause of the violation on the log. $a_d$ would also be an
explanation for the violation on $l$ if having this causal information
advances the epistemic knowledge of the agent. Note that there could
be traces arising from the initial configuration where the behavior
is inconsistent with the log. Knowing that $a_d$ is
consistent with the behavior on the log and that it is a cause of the
violation would advance the agent's knowledge and provide an explanation for the violation.

\subsection{Using Causality for Blame Attribution}
Actual causation is an integral part of the prominent theories of
blame in social psychology and legal
settings~\cite{shaver2012attribution, alicke2000culpable,
feinberg1985suaculpa, kenner1967blaming}. Most of these theories
provide a comprehensive framework for blame which integrates
causality, intentionality and
foreseeability~\cite{shaver2012attribution,
alicke2000culpable,lagnado2008judgments}. These theories recognize
blame and cause as interrelated yet distinct concepts. Prior to
attributing blame to an actor, a causal relation must be established
between the actor's actions and the outcome. However, not all actions
which are determined as a cause are blameworthy and an agent can be
blamed for an outcome even if their actions were not a direct cause
(for instance if an agent was responsible for another agent's
actions). In our work we focus on the first aspect where we develop a
theory for actual causation and provide a building block to find
blameworthy programs from this set.

We can use the causal set output by the definitions in
Section~\ref{sec:definitions} and further narrow down the set to find
blameworthy programs.  Note that in order to use our definition as a
building block for blame assignment, we require
information about a) which of the executed programs deviate from the
protocol, and b) which of these deviations are harmless. Some harmless
deviants might be output as part of the causal set because their
interaction is critical for the violation to occur. Definition~\ref{definition:cause3'} below provides one approach to
removing such non-blameworthy programs from the causal set. In
addition we can filter the norms from the causal
set. 

For this purpose, we use the notion of protocol specified norms ${\cal
N}$ introduced in Section~\ref{sec:application}.  We impose an additional
constraint on the norms, i.e., in the extreme counterfactual world
where we execute norms only, there should be no possibility of
violation. We call this condition \emph{necessity}. Conceptually,
necessity says that the reference standard (norms) we employ to assign
blame is reasonable.

\begin{definition}[Necessity condition for norms]\label{definition:necessity}
Given $\langle I, \Sigma, {\cal N}, \varphi_V \rangle$, we say that
${\cal N}$ satisfies the necessity condition w.r.t.\ 
$\varphi_{V}$ if for any trace $t'$ starting from the initial configuration $\langle I,
{\cal N}, \Sigma \rangle$, it is the case that
$t' \not \in \varphi_V$.
\end{definition}

We can use the norms ${\cal N}$ and the program cause $X$ with its
corresponding actual cause $a_d$ from Phase 2
(Definitions~\ref{definition:cause2a}, ~\ref{definition:cause2b}), in
order to determine whether a program is a harmless deviant as follows. Definition~\ref{definition:cause3'} presents a sound (but not complete) approach for identifying harmless deviants. 

\begin{definition}[Harmless deviant] \label{definition:cause3'}
Let $X$ be a program cause of violation $V$ and $a_d$ be the
corresponding actual cause as determined by
Definitions~\ref{definition:cause2a} and~\ref{definition:cause2b}.  We
say that the program corresponding to index $i \in X$ is a harmless
deviant w.r.t.\ trace $t$ and violation $\varphi_V$ if ${\cal A}(i)$
is deviant (i.e. ${\cal A}(i) \neq {\cal N}(i)$) and $a_d|_{i}$ is a
prefix of ${\cal N}(i)$.  \end{definition}

For instance in our case study (Section~\ref{sec:application}),
Theorem~\ref{theorem-passwords1} outputs $X$ and $a_d$
(Figure~\ref{fig:log1'}) as a cause. $X$
includes \Server1. Considering \Server1's norm
(Figure~\ref{fig:norms1'}), ${\cal A}\{\Server\}$ will be considered a
deviant, but according to
Definition~\ref{definition:cause3'}, \Server1\ will be classified as
a \emph{harmless deviant} because $a_d|_{\Server1}$ is a prefix of
${\cal N}(\Server1)$.
Note that in order to capture blame attribution accurately, we will
need a richer model which incorporates intentionality, epistemic
knowledge and foreseeability, beyond causality.


\section{Related Work} \label{sec:related} 

Currently, there are multiple proposals for providing accountability
in decentralized multi-agent systems~\cite{feigenbaum2011,
feigenbaum2011deterrence, kusters2010, backesDDMT06,
jagadeesanJPR09, haeberlenKD07, barth2007, gossler2010, wang2013}. 
Although the intrinsic relationship between causation and accountability is
often acknowledged, the foundational studies of accountability do not
explicitly incorporate the notion of cause in their formal definition or treat
it as a blackbox concept without explicitly defining it. Our thesis is that accountability is not a trace property since evidence from the log alone does not provide a justifiable basis to determine accountable parties. Actual causation is not a trace property; inferring actions which are actual causes of a violating trace requires analyzing counterfactual traces (see our sufficiency conditions).   Accountability depends on actual causation and is, therefore, also not a trace property. 

On the other hand, prior work on actual causation in analytical philosophy and AI has considered counterfactual based causation in detail~\cite{halpernpearl2001, halpernpearl2005, mackie1965, wright1985, hallbook, pearlbook2000}. These ideas have been applied for fault diagnosis where system components are analyzed, but these frameworks do not adequately capture all the elements crucial to model a security setting. Executions in security settings involve interactions among concurrently running programs in the presence of adversaries, and little can be assumed about the scheduling of events. We discuss below those lines of work which are most closely related to ours.

\paragraph{\textbf{Accountability}}
K{\"u}sters et al~\cite{kusters2010} define a protocol $P$ with
associated accountability constraints that are rules of the form: if a
particular property holds over runs of the protocol instances then
particular agents may be blamed.  Further, they define a judge $J$ who
gives a verdict over a run $r$ of an instance $\pi$ of a protocol $P$,
where the verdict blames agents. 
In their work, K{\"u}sters et al assume that the accountability
constraints for each protocol are given and complete. They state that
the judge $J$ should be designed so that $J$'s verdict is fair and
complete w.r.t. these accountability constraints.  They design a judge
separately for every protocol with a specific accountability
property. 
K{\"u}sters et al.'s definition of accountability has been successfully applied to
substantial protocols such as voting, auctions, and contract
signing. 
Our work complements this line of work
in that we aim to provide
a semantic basis for arriving at such accountability constraints,
thereby providing a justification for the blame assignment suggested
by those constraints.  Our actual cause definition can be viewed as a
generic judging procedure that is defined independent of the violation
and the protocol. 
We believe that using our cause definition as the basis for accountability
constraints would also ensure the minimality of verdicts given by the judges.

Backes et al~\cite{backesDDMT06} define accountability as the ability to show
evidence when an agent deviates. The authors analyze a
contract signing protocol using protocol composition logic.
In particular, the authors consider the case
when the trusted third-party acts dishonestly and prove that
the party can be held accountable by looking at a violating trace. 
This work can be viewed as a special case of the subsequent work of K{\"u}sters
et al.~\cite{kusters2010} where the property associated with the violating trace is an example 
of an accountability constraint.

Feigenbaum et al~\cite{feigenbaum2011,feigenbaum2011deterrence} also
propose a definition of accountability that focuses on linking a
violation to punishment. They use Halpern and Pearl's
definition~\cite{halpernpearl2001, halpernpearl2005} of causality in
order to define mediated punishment, where punishment is justified by
the existence of a causal chain of events in addition to satisfaction
of some utility conditions. 
The underlying ideas of our cause definition
could be adapted to their framework to instantiate the causality
notion that is currently used as a black box in their definition of
mediated punishment. One key difference is that we focus on finding program actions that lead to the violation, which could explain why the violation happened while they focus on establishing a causal chain between
violation and punishment events.

\paragraph{\textbf{Causation for blame assignment}}
The work by Barth et al~\cite{barth2007} provides a definition of
accountability that uses the much coarser notion of Lamport causality, which is related to Phase 1 of our definition. However, we use minimality checks and filter out \emph{progress
enablers} in Phase 2 to obtain a finer determination of actual cause.
 
G{\"o}ssler et al's work~\cite{gossler2010,gossler2013} 
considers blame assignment for safety property violations 
where the violation of the global safety property implies that some components have violated their local specifications. They use a counterfactual notion of causality similar in spirit to ours to identify a subset of these faulty components as causes of the violation. \cut{The blame assignment is done by using a single execution trace (i.e. the log) and the local specification of the components. } The most recent work in this line applies the framework to real-time systems specified using timed automata~\cite{gossler2014}.

A key technical difference between this line of work and ours is the way in which the contingencies to be considered in counterfactual reasoning are constructed. We have a program-based approach to leverage reasoning methods based on invariants and program logics.
G{\"o}ssler et al assume that a dependency relation that captures information flow between component actions are given and construct their contingencies using
the traces of faulty components observed on the log as a basis. A set of faulty components is the necessary cause of the violation if the violation would disappear once the traces of these faulty components are modified to match the components' local specifications. 
They determine the longest prefixes of faulty components that satisfy the specification and replace the faulty suffixes with a correct one. Doing such a replacement without taking into account its impact on the behavior of other components that interact with the faulty components would not be satisfactory. Indeed, Wang et al~\cite{wang2013} describe a counterexample to G{\"o}ssler et
al's work~\cite{gossler2010} where all causes are not found because of
not being able to completely capture the effect of one component's
behavior on another's.  The most recent definitions of  G{\"o}ssler et
al~\cite{gossler2013,gossler2014} address this issue by over approximating the parts of the log affected by the faulty components and replacing them with behavior that would have arisen had the faulty ones behaved correctly.

In constructing the contingencies to consider in counterfactual reasoning, we do not work with individual traces as  G{\"o}ssler et al. Instead, we 
work at the level of programs where ``correcting'' behavior is done by replacing program actions 
with those that do not have any effect on the violation other than enabling the programs to progress.
The relevant contingencies follow directly from the execution of programs where such replacements have been done, without any need to develop additional machinery for reconstructing traces. Note also that we have a sufficiently fine-grained definition to pinpoint the minimal set of actions that make the component a part of the cause, where these actions may a part be of faulty or non-faulty programs. 
\cut{Our approach also allows us to investigate causal impact of actions before the point where violation has been detected, not necessarily keeping prefixes before the violation fixed as in the work of G{\"o}ssler et al.
A basic premise of our work is that in the security setting parties in protocols deviate by exercising their choice to follow the prescribed norm or a different program, and this is best captured by replacing the entire program by a norm rather than suffixes of the individual faulty traces. Following a norm instead of deviating may give rise to changes in behavior not only after the point where the violation initially occurred but also before it. Our causal analysis accounts for such scenarios. Note that this distinction would still exist if we generalized our framework to causal anaylses where programs are considered as black-boxes.} 
Moreover, we purposely separate cause determination and blame assignment because we believe that in the security setting, blame assignment is a problem that requires additional criteria to be considered such as the ability to make a choice, and intention. The work presented in this paper focuses on identifying cause as a \emph{building block} for blame assignment. 


\section{Conclusion} \label{sec:conclusion} 
We have presented a first attempt at defining what it means for a sequence
of program actions to be an actual cause of a violation of a security
property. This question is motivated by security applications where
agents can exercise their choice to either execute a prescribed
program or deviate from it. While we demonstrate the value of this
definition by analyzing a set of authentication failures, it would be
interesting to explore applications to other protocols in which
accountability concerns are central, in particular, protocols for
electronic voting and secure multiparty computation in the semi-honest
model.
Another challenge in security settings is that deviant programs
executed by malicious agents may not be available for analysis; rather
there will be evidence about certain actions committed by such
agents. A generalized treatment accounting for such partial
observability would be technically interesting and useful for
other practical applications. This work demonstrates the importance of program actions as causes as a useful \emph{building block} for several such applications, in particular for providing explanations, assigning blame and providing accountability guarantees for security protocols.



\bibliographystyle{IEEEtran}
{\bibliography{arxiv-cause-actions}}

\newpage
\onecolumn
\section*{Appendix} \label{sec:appendix} 

\subsection{Operational Semantics}

Selected rules of the operational semantics of the programming language $L$ are shown below.
\\

\begin{small}
\noindent \framebox{$T \tstep T'$}
\begin{mathpar}

\inferrule{\pred{eval}~t~t'} {I \wsep t \wsep \sigma
  \tstept{\epsilon} I \wsep t' \wsep \sigma'}\rname{red-end}\and

\inferrule{\sigma; \zeta(t) \astep \sigma'; t' \\ \pred{eval}~t'~t''}
 {I \wsep ((b: x = \zeta(t)); e) \wsep \sigma
    \tstept{\langle I, b\rangle} I \wsep e \{t''/x\} \wsep
    \sigma'}\rname{red-act}\and

\inferrule{\pred{eval}~t~\btrue} {I \wsep (\action{assert}(t); e) \wsep
  \sigma \tstept{\epsilon} I \wsep e \wsep
  \sigma}\rname{red-assert}\and

\end{mathpar}
\framebox{${\cal C} \cstep {\cal C'}$}

\begin{flushleft}
\emph{Internal reduction}
\end{flushleft}
\begin{mathpar}
\inferrule{T_i \tstept{r} T_i'}{\ldots, T_i,\ldots \cstept{r}
  \ldots,T_i',\ldots}\rname{red-config}
\end{mathpar}

\begin{flushleft}\emph{Communication action}\end{flushleft}
\begin{mathpar}
\inferrule{\pred{eval}~t~t'}{\ldots,\langle I_s\wsep ((b_s:
  x=\action{send}(t)); e_s) \wsep \sigma_s \rangle, \langle I_r\wsep
  ((b_r: y=\action{recv}()); e_r) \wsep \sigma_r \rangle, \ldots
  \\\\ \cstept{\langle\langle I_s, b_s \rangle, \langle I_r, b_r
    \rangle \rangle} \ldots, \langle I_s\wsep e_s[0/x]\wsep
  \sigma_s\rangle, \langle I_r \wsep e_r[t'/y] \wsep \sigma_r \rangle,
  \ldots}\rname{red-comm}
\end{mathpar}
\end{small}

\subsection{Case study: Compromised notaries attack}

We model an instance of our running example based on
passwords in order to demonstrate our actual cause definition. As
explained in Section~\ref{sec:example}, we consider a protocol session
where \Server1, \User1, \User2, \User3\ and multiple notaries interact over an
adversarial network to establish access over a password-protected
account. In parallel for this scenario, we assume the log also contains interactions of a second server (\Server2), one notary (\Notary4, not contacted by \User1, \User2 or \User3) 
and another user (\User4) who follow their norms for account access. These threads do not interact with threads \{\User1, \Server1, \Notary1, \Notary2, 
\Notary3, \Adversary, \User2, \User3\}. The protocol has been described in detail below.

\subsubsection{Protocol Description}
\label{subsec:protocol1}
We consider our example protocol with eleven threads
named \{\Server1, \User1, \User2, \User3, \Adversary, \Notary1, \Notary2, \Notary3, \Notary4, \Server2, \User4\}. The \emph{norms} for all these threads, except {\Adversary}
are shown in Figure~\ref{fig:norms1}. The
actual violation is caused because some of the executing programs are
different from the norms. These actual programs, called ${\cal A}$ as
in Section~\ref{sec:definitions}, are shown later. The norms are shown
here to help the reader understand what the ideal protocol is. 

In this case study, we have two servers (\Server1, \Server2) running the protocol with two different users (\User1, \User4) and each server allocates 
account access separately. The norms in Figure~\ref{fig:norms1} assume that {\User1}'s and {\User4}'s accounts
(called $acct_1$ and $acct_2$ in {\Server1}'s and {\Server2}'s norm respectively) have been created already. \User1's password, $pwd_1$ is 
associated with {\User1}'s user id $uid1$. Similarly \User4's password $pwd_{2}$ is associated with its user id $uid2$.  This association (in hashed form) is stored in {\Server1}'s
local state at pointer $mem_1$ (and at $mem_2$ for \Server2). The norm for {\Server1} is to wait for a
request from an entity, respond with its public key, then wait for a
password encrypted with that public key and grant access to the
requester if the password matches the previously stored value
in \Server1's memory at $mem_1$. To grant access, \Server1 adds an
entry into a private access matrix, called $P_1$. (A separate server
thread, not shown here, allows {\User1} to access its resource if this
entry exists in $P_1$.)

The norm for {\User1} is to send an access request to {\Server1}, wait
for the server's public key, verify that key with three notaries and
then send its password $pwd_{1}$ to {\Server1}, encrypted under \Server1's
public key. On receiving \Server1's public key, {\User1} initiates a
protocol with the three notaries and accepts or rejects the key based
on the response of a majority of the notaries.  

The norm for \User4 is the same as that for \User1 except that it interacts with \Server2. Note that \User4 only verifies the public key with one notary, \Notary4. The norm for \Server2 is the same as that for \Server1 except that it interacts with \User4. 

In
parallel, the norm for \User2\ is to generate and send a nonce to \User3.  The norm for \User3 is to receive a message from \User2,
generate a nonce and send it to \User2. 

Each notary has a private database of \textit{(public\_key,
principal)} tuples. The norms here assume that this database has
already been created correctly. When {\User1} or {\User4} send a request with a
public key, the notary responds with the principal's identifier after
retrieving the tuple corresponding to the key in its database. (Note
that, in this simple example, we identify threads with principals, so
the notaries just store an association between public keys and their
threads.) 

\begin{figure}

\small

\begin{tabular}{p{0.5\textwidth}}

				\framebox{\noindent\textbf{Norm ${\cal N}(\Server1)$:}}\\
		\begin{tabular}{l}
		$
		\begin{array}{l}

   		 1: (uid1, n1) = \action{recv}(j);  ~~// \text{access req from thread $j$}\\
			2: n2 = \action{new}; \\
		 3: \action{send}(j, (pub\_key\_\Server1, n2, n1));~~// \text{sign and send public key}\\  
		4: s1 = \action{recv}(j);  ~~// \text{encrypted $uid1,pwd1$ from $j$, alongwith its thread id $J$}\\
    5: (n3, uid1, pwd1, J) = \pred{Dec}(pvt\_key\_\Server1, s1);\\ 		  
		6: t = \pred{Hash}(uid1, pwd1); \\
		\mbox{  }~~ \action{assert} (mem_1 = t)  ~~// \text{compare hash with stored hash value for same uid}\\
		7: \action{insert}(P_1, (acct_1, J)); \\
		\end{array} $\\
		\end{tabular}
		
			\\

		\framebox{\noindent\textbf{Norm ${\cal N}(\User1)$:}}\\
		\begin{tabular}{l}
		$
		\begin{array}{l}

		1: n1 = \action{new}; \\
		2: \action{send}(\Server1, (uid1, n1)); ~~// \text{access request}\\
		3: (pub\_key1, n2, n1) = \action{recv}(j);  ~~//\text{key from $j$}\\
		4: n3, n4, n5 = \action{new}; \\
		5: \action{send}(\Notary1, pub\_key1, n3);\\
		6: \action{send}(\Notary2, pub\_key1, n4);\\
		7: \action{send}(\Notary3, pub\_key1, n5);\\
	  8: \pred{Sig}(pvt\_key\_\Notary1, (pub\_key1, l1,n3)) = \action{recv}(\Notary1);~~//\text{notary1 responds }\\
	        9: \pred{Sig}(pvt\_key\_\Notary2, (pub\_key1, l2,n4)) = \action{recv}(\Notary2);~~//\text{notary2 responds }\\
					10: \pred{Sig}(pvt\_key\_\Notary3, (pub\_key1, l3,n5)) = \action{recv}(\Notary3);~~//\text{notary3 responds }\\
		\mbox{  }~~\action{assert} (\text{At least two of \{l1,l2,l3\} equal \Server1}) \\
		11: t = \pred{Enc}(pub\_key1, n2, (uid1, pwd1, \User1 )); \\
		12: \action{send}(\Server1, t); ~~// \text{send $t$ to \Server1};\\
		
			\end{array} $ \\
			\end{tabular}\\
			
			\\
			
				\framebox{\noindent\textbf{Norms ${\cal N}(\Notary1), {\cal N}(\Notary2), {\cal N}(\Notary3), {\cal N}(\Notary4)$:}}\\

		\begin{tabular}{l}
		$
		\begin{array}{l}
		//\text{ $o$ denotes \Notary1, \Notary2, \Notary3 or \Notary4}\\
		1: (pub\_key, n1) = \action{recv}(j); \\
		2: pr = \pred{KeyOwner}(pub\_key); ~~~// \text{lookup key owner} \\
		3: \action{send}(j, \pred{Sig}(pvt\_key\_o, (pub\_key, pr, n1))); ~~// \text{signed certificate}; \\

		\end{array} $ 

		\end{tabular}\\

					\framebox{\noindent\textbf{Norm ${\cal N}(\Server2)$:}}\\
		\begin{tabular}{l}
		$
		\begin{array}{l}

   		 1: (uid2, n1) = \action{recv}(j);  ~~// \text{access req from thread $j$}\\
			2: n2 = \action{new}; \\
		 3: \action{send}(j, (pub\_key\_\Server2, n2, n1));\\ 
		4: s1 = \action{recv}(j);  ~~// \text{encrypted $uid2,pwd2$ from $j$, alongwith its thread id $J$}\\
    5: (n2, uid2, pwd2, J) = \pred{Dec}(pvt\_key\_\Server2, s1);\\ 		  
		6: t = \pred{Hash}(uid2, pwd2); \\
		\mbox{  }~~ \action{assert} (mem_2 = t)  ~~// \text{compare hash with stored hash value for same uid}\\
		7: \action{insert}(P_2, (acct_2, J)); \\
   		 		\end{array} $
		\end{tabular}
		
			\\

	\framebox{\noindent\textbf{Norm ${\cal N}(\User4)$:}}\\		\begin{tabular}{l}
		$
		\begin{array}{l}

		1: n1 = \action{new}; \\
		2: \action{send}(\Server2, (uid2, n1)); ~~// \text{access request}\\
		3: pub\_key, n2, n1 = \action{recv}(j);  ~~//\text{key from $j$}\\
		4: n3 = \action{new}; \\
		5: \action{send}(\Notary4, pub\_key, n3);\\
		6: \pred{Sig}(pvt\_key\_\Notary4, (pub\_key, l1,n3)) = \action{recv}(\Notary4);~~//\text{notary4 responds }\\
		\mbox{  }~~ \action{assert} (\text{\{l1\} equals \Server2}) \\
		7: t = \pred{Enc}(pub\_key, n2, (uid2, pwd2, \User4 )); \\
		8: \action{send}(\Server2, t); ~~// \text{send $t$ to \Server2};\\
		
				\end{array} $ 
                \end{tabular}
		\\
	
	\framebox{\noindent\textbf{Norm ${\cal N}(\User2)$:}}\\		\begin{tabular}{l}
		$
		\begin{array}{l}
		1: n1 = \action{new}; \\
		2: \action{send}(\User3, (n1)); \\
		3: \pred{Sig}(pvt\_key\_j, (n2, n1)) = \action{recv}(\User3);
		4: 
		\end{array} $ 
                \end{tabular}
		\\

	\framebox{\noindent\textbf{Norm ${\cal N}(\User3)$:}}\\		
	\begin{tabular}{l}
		$
		\begin{array}{l}
		1:  n1 = \action{recv}(\User2);\\
		2:  n2 = \action{new}; \\
		3: \action{send}(\User3, \pred{Sig}(pvt\_key\_\User3, (n2, n1)) ); \\
		\end{array} $ 
                \end{tabular}
\end{tabular}

\caption{Norms for \Server1, \User1, \Server2, \User4, \User2, \User3 and the notaries. \Adversary's norm is the trivial empty program.}
\label{fig:norms1}
\end{figure}

\subsubsection{Preliminaries}


\paragraph{Notation} 
The programs in this example use several primitive functions
$\zeta$. $\pred{Enc}(k,m)$ and $\pred{Dec}(k',m)$ denote encryption
and decryption of message $m$ with key $k$ and $k'$
respectively. $\pred{Hash}(m)$ generates the hash of term
$m$. $\pred{Sig}(k,m)$ denotes message $m$ signed with the key $k$,
paired with $m$ in the clear. $pub\_key\_i$ and $pvt\_key\_i$ denote
the public and private keys of thread $i$, respectively. For
readability, we include the intended recipient $i$ and expected sender
$j$ of a message as the first argument of $\action{send}(i,m)$ and
$\action{recv}(j)$ expressions. As explained earlier, $i$ and $j$ are
ignored during execution and a network adversary, if present, may
capture or inject any messages. $\pred{Send}(i, j, m) \at u$ holds if thread $i$ sends message $m$ to thread $j$ at time $u$ and $\pred{Recv}(i, j, m) \at u$ hold if thread $i$ receives message $m$ from thread $j$ at time $u$. $P_1(u)$ and $P_2(u)$ denotes the tuples in the permission matrices at time $u$. Initially $P_1$ and $P_2$ do not contain any access 
permissions.

\paragraph{Assumptions} 
(A1) \[\pred{HonestThread}(\Server1,{\cal A}(\Server1))\] 
We are interested in security guarantees about users who create accounts by interacting with the server and who do not share the generated password 
or user-id with any other principal except for sending it according to the roles specified in the program given below. 

(A2) \[\pred{HonestThread}(\User1,{\cal A}(\User1))\]
(A3) \[\pred{HonestThread}(\Adversary, {\cal A}(\Adversary))\]
(A4) \[\pred{HonestThread}(\Notary1, {\cal A}(\Notary1))\]
(A5) \[\pred{HonestThread}(\Notary2, {\cal A}(\Notary2))\]
(A6) \[\pred{HonestThread}(\Notary3,{\cal A}(\Notary3))\]
(A7) \[\pred{HonestThread}(\Notary4,{\cal A}(\Notary4))\]
(A8) \[\pred{HonestThread}(\Server2,{\cal A}(\Server2))\]
(A9) \[\pred{HonestThread}(\User4,{\cal A}(\User4))\]
(A10) \[\pred{HonestThread}(\User2,{\cal A}(\User2))\]
(A11) \[\pred{HonestThread}(\User3,{\cal A}(\User3))\]

A principal following the protocol never shares its keys with any other entity. We also assume that the encryption scheme in semantically secure and non-malleable. 
Since we identify threads with principals therefore each of the threads are owned by principals with the same identifier, for instance \Server1 owns the thread that executes the program {\cal A}(\Server1).

(Start1) \[Start(i) \at \ninfty \] where $i$ refers to all the threads in the set described above. 

\paragraph{Security property}
The security property of interest to us is that if at time $u$, a
thread $k$ is given access to account $a$, then $k$ owns
$a$. Specifically, in this example, we are interested in the $a =
acct_1$ and $k = \User1$. This can be formalized by the following
logical formula, $\neg \varphi_V$:
\begin{equation}\label{eqn:property'}
 \forall u,k. ~(acct_1, k) \in
P_1(u) \imp (k = \User1)
\end{equation}
Here, $P_1(u)$ is the state of the access control matrix $P_1$ for
\Server1 at time $u$.

The actuals for all threads are shown in Figure~\ref{fig:actuals1} and~\ref{fig:actuals2}.

	\begin{figure}

\small

		\begin{tabular}{p{0.5\textwidth}}

	\framebox{\noindent\textbf{Actual ${\cal A}(\Adversary)$ }}\\

		\begin{tabular}{l}
		$
		\begin{array}{l}
		1: (uid1, n1) = \action{recv}(j);     ~~// \text{intercept req from $\User1$}\\
		2: n2 = \action{new}; \\
		3: \action{send}(\User1, (pub\_key\_\Adversary1, n2, n1)); ~~// \text{send key to \User1}\\ 
		4: s = \action{recv}(\User1); ~~// \text{pwd from User}\\
		5: n2, uid1, pwd1, \User1 = \pred{Dec}(pvt\_key\_\Adversary, s); ~~// \text{decrypt pwd}; \\
		6: n3 = \action{new};\\
		7: \action{send}(\Server1, (uid1, n3)); ~~// \text{access request to \Server}\\
		8: pub\_key, n4, n3 = \action{recv}(\Server1); \\
		9: t = \pred{Enc}(pub\_key, (n4, uid1, pwd1, \Adversary)); // \text{encrypt pwd}\\
		10: \action{send}(\Server1, t); ~~// \text{pwd to \Server1} \\
		\end{array} $ \\
		\end{tabular}\\
		\\
		
		\framebox{\noindent\textbf{Actuals ${\cal A}(\Notary1), {\cal A}(\Notary2), {\cal A}(\Notary3)$:}}\\

		\begin{tabular}{l}
		$
		\begin{array}{l}
			//\text{ $o$ denotes \Notary1, \Notary2 or \Notary3}\\
	1:(pub\_key\_\Adversary, n1) = \action{recv}(j); \\
		2:\action{send}(j, \pred{Sig}(pvt\_key\_o, (pub\_key\_\Adversary, \Server1, n1));  ~~// \text{signed certificate to $j$}; \\%
		\end{array} $ 
		\end{tabular}\\
               \\

	\framebox{\noindent\textbf{Actual ${\cal A}(\Server1)$:}}\\
		\begin{tabular}{l}
		$
		\begin{array}{l}
   		1: (uid1, n1) = \action{recv}(j);  ~~// \text{access req from thread $j$}\\
		2: n2 = \action{new}; \\
		3: \action{send}(j, (pub\_key\_\Server1, n2, n1));\\ 
		4: n4 = \action{recv}(j);  ~~// \text{receive nonce from thread \User2}\\
		5: n5 = \action{new}; \\
		6: \action{send}(j, \pred{Sig}(pvt\_key\_\Server1, (n5, n4)));\\ 
		7: s1 = \action{recv}(j);  ~~// \text{encrypted $uid1,pwd1$ from $j$, alongwith its thread id $J$}\\
                 8: (n3, uid1, pwd1, J) = \pred{Dec}(pvt\_key\_\Server1, s1);\\ 		  
		
		9: t = \pred{Hash}(uid1, pwd1); \\
		\mbox{  }~~\action{assert} (mem_1 = t) \textbf{[A]} ~~//  \text{compare hash with stored hash value for same uid}\\
		10: \action{insert}(P_1, (acct_1, J)); \\

		\end{array} $\\
		\end{tabular}\\

	\\
		\framebox{\noindent\textbf{Actual ${\cal A}(\User1)$:}}\\
		\begin{tabular}{l}
		$
		\begin{array}{l}

		1: n1 = \action{new}; \\
		2: \action{send}(\Server1, (uid1, n1)); ~~// \text{access request}\\
		3: (pub\_key, n2, n1) = \action{recv}(j);  ~~//\text{key from $j$}\\
		4: n3, n4, n5 = \action{new}; \\
		5: \action{send}(\Notary1, pub\_key, n3);\\
		6: \action{send}(\Notary2, pub\_key, n4);\\
		7: \action{send}(\Notary3, pub\_key, n5);\\
	  8: \pred{Sig}(pvt\_key\_\Notary1, (pub\_key, l1,n3)) = \action{recv}(\Notary1);~~//\text{notary1 responds }\\
	        9: \pred{Sig}(pvt\_key\_\Notary2, (pub\_key, l2,n4)) = \action{recv}(\Notary2);~~//\text{notary2 responds }\\
					10: \pred{Sig}(pvt\_key\_\Notary3, (pub\_key, l3,n5)) = \action{recv}(\Notary3);~~//\text{notary3 responds }\\
		\mbox{  }~~ \action{assert} (\text{At least two of} \{l1,l2,l3\} \text{equal \Server1});\textbf{[B]}~~// \\
		11: t = \pred{Enc}(pub\_key, n2, (uid1, pwd1, \User1 )); \\
		12: \action{send}(\Server1, t); ~~// \text{send $t$ to \Server1};\\
		
			\end{array} $ \\
			\end{tabular}\\
	\\

	\framebox{\noindent\textbf{Actual ${\cal A}(\User2)$:}}\\		
	\begin{tabular}{l}
		$
		\begin{array}{l}
		1: n1 = \action{new}; \\
		2: \action{send}(\Server1, (n1)); \\
		3: \pred{Sig}(pvt\_key,(n2, n1)) = \action{recv}(\Server1);\\
		4: \action{send}(\User3, (n2));  \\
		5: \pred{Sig}(pub\_key, n3, n2) = \action{recv}(\User3);
		\end{array} $ 
                \end{tabular}\\
		\\

	\framebox{\noindent\textbf{Actual ${\cal A}(\User3)$:}}\\		
	\begin{tabular}{l}
		$
		\begin{array}{l}
		1:  n1 = \action{recv}(\User2);\\
		2:  n2 = \action{new}; \\
		3: \action{send}(\User3, \pred{Sig}(pvt\_key\_\User3, n2, n1) ); \\
		\end{array} $ 
                \end{tabular}\\
	
\end{tabular}

\caption{Actuals for \Adversary, \Notary1, \Notary2, \Notary3, \Server1, \User1, \User2, \User3}
\label{fig:actuals1}
\end{figure}

		
		\begin{figure}

\small

		\begin{tabular}{p{0.5\textwidth}}

								\framebox{\noindent\textbf{Actual ${\cal A}(\Server2)$:}}\\
		\begin{tabular}{l}
		$
		\begin{array}{l}

   		 1: (uid2, n1) = \action{recv}(j);  ~~// \text{access req from thread $j$}\\
			2: n2 = \action{new}; \\
		 3: \action{send}(j,(pub\_key\_\Server2, n2, n1));\\ 
		4: s1 = \action{recv}(j);  ~~// \text{encrypted $uid2,pwd2$ from $j$, alongwith its thread id $J$}\\
    5: (n2, uid2, pwd2, J) = \pred{Dec}(pvt\_key\_\Server2, s1);\\ 		  
		6: t = \pred{Hash}(uid2, pwd2); \\
		\mbox{  }~~\action{assert} (mem_2 = t)  ~~//(C) \text{compare hash with stored hash value for same uid}\\
		7: \action{insert}(P_2, (acct_2, J)); \\
   		 		\end{array} $\\
		\end{tabular}\\
		
			\\

	\framebox{\noindent\textbf{Actual ${\cal A}(\User4)$:}}\\		
	\begin{tabular}{l}
		$
		\begin{array}{l}

		1: n1 = \action{new}; \\
		2: \action{send}(\Server2, (uid2, n1)); ~~// \text{access request}\\
		3: \pred{Sig}(pub\_key, n2, n1) = \action{recv}(j);  ~~//\text{key from $j$}\\
		4: n3 = \action{new}; \\
		5: \action{send}(\Notary4, pub\_key, n3);\\
		6: \pred{Sig}(pvt\_key\_\Notary4, (pub\_key, l1,n3)) = \action{recv}(\Notary4);~~//\text{notary4 responds }\\
		\mbox{  }~~  \action{assert} (\text{\{l1\} equals \Server2}) (D)\\
		7: t = \pred{Enc}(pub\_key, n2, (uid2, pwd2, \User4 )); \\%
		8: \action{send}(\Server2, t); ~~// \text{send $t$ to \Server2};\\
		
				\end{array} $ 
                \end{tabular}\\
		\\

				\framebox{\noindent\textbf{Actual ${\cal A}(\Notary4)$:}}\\

		\begin{tabular}{l}
		$
		\begin{array}{l}
		//\text{ $o$ denotes \Notary1, \Notary2, \Notary3 or \Notary4}\\
		1: (pub\_key, n1) = \action{recv}(j); \\
		2: pr = \pred{KeyOwner}(pub\_key); ~~~// \text{lookup key owner} \\
		3: \action{send}(j, \pred{Sig}(pvt\_key\_o, (pub\_key, pr, n1))); ~~// \text{signed certificate}; \\

		\end{array} $ 

		\end{tabular}\\

	\end{tabular}

\caption{Actuals for \Server2, \User4, \Notary4}
\label{fig:actuals2}
\end{figure}

\subsubsection{Attack}
As an illustration, we model the ``Compromised Notaries'' violation of Section~\ref{sec:example}. The programs executed by all threads are given in Figures~\ref{fig:actuals1} and~\ref{fig:actuals2}. \User1\ sends an access
request to \Server1\ which is intercepted by \Adversary\ who sends its
own key to \User1\ (pretending to be \Server1). \User1\ checks with
the three notaries who falsely verify \Adversary's public key to be
\Server1's key. Consequently, \User1\ sends the password to
\Adversary. {\Adversary}\ then initiates a protocol with {\Server1}
and gains access to the \User1's account. Note that the actual
programs of the three notaries attest that the public key given to
them belongs to \Server1. In parallel, \User2\ sends a request to \Server1 and receives a response from \Server1. Following this interaction, \User2 interacts with \User3, as in their norms. \User4, \Server2 and \Notary4 execute
their actuals in order to access the account $acct_2$ as well.

Figure~\ref{fig:log1} shows the expressions executed by each thread
on the property-violating trace. For instance, the label
$\langle \langle \User1, 1 \rangle, \langle \Adversary,
1 \rangle \rangle $ indicates that both \User1 and \Adversary\
executed the expressions with the line number 1 in their actual
programs, which resulted in a synchronous communication between them,
while the label $\langle \Adversary, 4 \rangle $ indicates the local
execution of the expression at line~4 of \Adversary's program. The
initial configuration has the programs: $\{{\cal A}(\User1), {\cal
A}(\Server1), {\cal A}({\Adversary}), {\cal A}({\Notary1}),  {\cal
A}({\Notary2}), {\cal A}({\Notary3}), {\cal A}(\User2), {\cal
A}(\User3), {\cal A}(\User4), {\cal A}(\Server2),\\{\cal A}(\Notary4) \}$. 
For this attack scenario, the concrete trace $t$ we consider is such that $\log(t)$ is any
\emph{arbitrary interleaving} of the actions for $X_1= \{\Adversary, \User1,
\Server1, \Notary1, \Notary2, \Notary3, \User2, \User3\}$ and $X_2=\{\Server2, \User4,
\Notary4\}$ shown in Figure~\ref{fig:log1}(a) and Figure~\ref{fig:log2}. Any such interleaved log is
denoted $\log(t)$ in the sequel. At the end of this log, $(acct_{1}, \Adversary)$ occurs in the access
control matrix $P_1$, but {\Adversary} does not own $acct_{1}$. Hence, this
log corresponds to a violation of our security property.


\begin{figure*}

\small
\begin{tabular}{p{0.33\textwidth}|p{0.33\textwidth}|p{0.34\textwidth}}

\begin{tabular}[t]{l}

\small

	{\noindent\textbf{(a)}}\\

	\framebox{\noindent\textbf{$\log(t)|_{\Adversary}$ }}\\

		\begin{tabular}{l}
		$
		\begin{array}{l}
		
		\langle \langle \User1, 2 \rangle, \langle \Adversary, 1 \rangle \rangle, \\
			\langle  \Adversary, 2 \rangle,\\
	
			\langle \langle \Adversary, 3  \rangle, \langle \User1, 3 \rangle \rangle, \\
			\langle \langle \User1, 12 \rangle, \langle \Adversary, 4 \rangle \rangle,\\
			\langle  \Adversary, 5 \rangle,\\
			\langle  \Adversary, 6 \rangle,\\
			
		       \langle \langle \Adversary, 7 \rangle,\langle \Server1, 1 \rangle \rangle,\\
			\langle \langle \Server1, 3 \rangle, \langle \Adversary, 8 \rangle \rangle, \\
			\langle  \Adversary, 9 \rangle,\\
			
      		\langle \langle \Adversary, 10  \rangle,  \langle \Server1, 7 \rangle \rangle, \\

		\end{array} $ \\
		\end{tabular}\\
		\\
		
		\framebox{\noindent\textbf{$\log(t)|_{\User1}$}}\\
		\begin{tabular}{l}
		$
		\begin{array}{l}
		\langle \User1, 1 \rangle,\\
		\langle \langle \User1, 2 \rangle, \langle \Adversary, 1 \rangle \rangle, \\
		\langle \langle \Adversary, 3 \rangle, \langle \User1, 3 \rangle \rangle, \\
		\langle \User1, 4 \rangle,\\
		\langle \langle \User1, 5 \rangle,  \langle \Notary1, 1 \rangle \rangle, \\
		\langle \langle \User1, 6 \rangle,  \langle \Notary2, 1 \rangle \rangle,\\
		\langle \langle \User1, 7 \rangle,  \langle \Notary3, 1 \rangle \rangle,\\

		   \langle \langle \Notary1, 2 \rangle, \langle \User1, 8 \rangle \rangle,  \\
                  \langle \langle \Notary2, 2 \rangle, \langle \User1, 9 \rangle \rangle,  \\
                          \langle \langle \Notary3, 2 \rangle, \langle \User1, 10 \rangle \rangle,  \\
	 		\langle \User1, 11 \rangle, \\
			\langle \langle \User1, 12 \rangle, \langle \Adversary, 4 \rangle \rangle,\\
		
			\end{array} $ \\
			\end{tabular}\\
			
\\
		\framebox{\noindent\textbf{$\log(t)|_{\Server1}$:}}\\
		\begin{tabular}{l}
		$
		\begin{array}{l}

		  \langle \langle \Adversary, 7 \rangle,\langle \Server1, 1 \rangle \rangle,\\
   				\langle \Server1, 2 \rangle, \\
				\langle \langle \Server1, 3 \rangle, \langle \Adversary, 8 \rangle \rangle, \\
				    \langle  \langle \User2, 2 \rangle, \langle \Server1, 4 \rangle \rangle, \\
						\langle \Server1, 5 \rangle, \\
		\langle \langle \Server1, 6 \rangle, \langle \User2, 3 \rangle \rangle, \\
				\langle \langle \Adversary, 10 \rangle,  \langle \Server1, 7 \rangle \rangle, \\
               \langle \Server1, 8 \rangle,\\ 		  
		       \langle \Server1, 9 \rangle,\\ 	
			       \langle \Server1, 10 \rangle,\\ 
		 				
		\end{array} $\\
		\end{tabular}\\
		\\
			
		\framebox{\noindent\textbf{$\log(t)|_{\Notary1}$:}}\\

		\begin{tabular}{l}
		$
		\begin{array}{l}
		\langle \langle \User1, 5 \rangle,  \langle \Notary1, 1 \rangle \rangle, \\
                  \langle \langle \Notary1, 2 \rangle, \langle \User1, 8 \rangle \rangle, \\
		\end{array} $ 
		\end{tabular}\\
               \\

		\framebox{\noindent\textbf{$\log(t)|_{\Notary2}$:}}\\

		\begin{tabular}{l}
		$
		\begin{array}{l}
		\langle \langle \User1, 6 \rangle,  \langle \Notary2, 1 \rangle \rangle, \\
                	\langle \langle \Notary2, 2 \rangle, \langle \User1, 9 \rangle \rangle,  \\
		\end{array} $ 
		\end{tabular}\\
               \\
               
               \framebox{\noindent\textbf{$\log(t)|_{\Notary3}$:}}\\

		\begin{tabular}{l}
		$
		\begin{array}{l}
		\langle \langle \User1, 7 \rangle,  \langle \Notary3, 1 \rangle \rangle, \\
                	\langle \langle \Notary3, 2 \rangle, \langle \User1, 10 \rangle \rangle,  \\
		\end{array} $ 
		\end{tabular}\\
               \\

		   \framebox{\noindent\textbf{$\log(t)|_{\User2}$:}}\\

		\begin{tabular}{l}
		$
		\begin{array}{l}
		    \langle \User2, 1 \rangle, \\
				\langle  \langle \User2, 2 \rangle, \langle \Server1, 4 \rangle \rangle, \\
		\langle \langle \Server1, 6 \rangle, \langle \User2, 3 \rangle \rangle, \\
		 \langle \langle \User2, 4 \rangle, \langle \User3, 1 \rangle \rangle, \\
			 \langle  \langle \User3, 3 \rangle, \langle \User2, 5 \rangle \rangle, \\
		\end{array} $ 
		\end{tabular}\\
               \\
               
                 \framebox{\noindent\textbf{$\log(t)|_{\User3}$:}}\\

		\begin{tabular}{l}
		$
		\begin{array}{l}
		 \langle \langle \User2, 4 \rangle, \langle \User3, 1 \rangle \rangle, \\
		\langle \User3, 2 \rangle, \\
		 \langle  \langle \User3, 3 \rangle, \langle \User2, 5 \rangle \rangle, \\
		\end{array} $ 
		\end{tabular}\\
               
							\end{tabular}
							
						&
\begin{tabular}[t]{l}

\small

	{\noindent\textbf{(b)}}\\
							\framebox{\noindent\textbf{$l|_{\Adversary}$ }}\\

		\begin{tabular}{l}
		$
		\begin{array}{l}
		
		\langle \langle \User1, 2 \rangle, \langle \Adversary, 1 \rangle \rangle, \\
			\langle  \Adversary, 2 \rangle,\\
	
			\langle \langle \Adversary, 3  \rangle, \langle \User1, 3 \rangle \rangle, \\
			\langle \langle \User1, 12 \rangle, \langle \Adversary, 4 \rangle \rangle,\\
			\langle  \Adversary, 5 \rangle,\\
			\langle  \Adversary, 6 \rangle,\\
			
		       \langle \langle \Adversary, 7 \rangle,\langle \Server1, 1 \rangle \rangle,\\
			\langle \langle \Server1, 3 \rangle, \langle \Adversary, 8 \rangle \rangle, \\
			\langle  \Adversary, 9 \rangle,\\
			
      		\langle \langle \Adversary, 10  \rangle,  \langle \Server1, 7 \rangle \rangle, \\

		\end{array} $ \\
		\end{tabular}\\
		\\
		
		\framebox{\noindent\textbf{$l|_{\User1}$}}\\
		\begin{tabular}{l}
		$
		\begin{array}{l}
		\langle \User1, 1 \rangle\\
		\langle \langle \User1, 2 \rangle, \langle \Adversary, 1 \rangle \rangle, \\
		\langle \langle \Adversary, 3 \rangle, \langle \User1, 3 \rangle \rangle, \\
		\langle \User1, 4 \rangle,\\
		\langle \langle \User1, 5 \rangle,  \langle \Notary1, 1 \rangle \rangle, \\
		\langle \langle \User1, 6 \rangle,  \langle \Notary2, 1 \rangle \rangle,\\
		\langle \langle \User1, 7 \rangle,  \langle \Notary3, 1 \rangle \rangle,\\

		   \langle \langle \Notary1, 2 \rangle, \langle \User1, 8 \rangle \rangle,  \\
                  \langle \langle \Notary2, 2 \rangle, \langle \User1, 9 \rangle \rangle,  \\
                          \langle \langle \Notary3, 2 \rangle, \langle \User1, 10 \rangle \rangle,  \\
	 		\langle \User1, 11 \rangle, \\
			\langle \langle \User1, 12 \rangle, \langle \Adversary, 4 \rangle \rangle,\\
		
			\end{array} $ \\
			\end{tabular}\\
			
\\
		\framebox{\noindent\textbf{$l|_{\Server1}$:}}\\
		\begin{tabular}{l}
		$
		\begin{array}{l}

		  \langle \langle \Adversary, 7 \rangle,\langle \Server1, 1 \rangle \rangle,\\
   				\langle \Server1, 2 \rangle, \\
				\langle \langle \Server1, 3 \rangle, \langle \Adversary, 8 \rangle \rangle, \\
				    \langle  \langle \User2, 2 \rangle, \langle \Server1, 4 \rangle \rangle, \\
						\langle \Server1, 5 \rangle, \\
		\langle \langle \Server1, 6 \rangle, \langle \User2, 3 \rangle \rangle, \\
				\langle \langle \Adversary, 10 \rangle,  \langle \Server1, 7 \rangle \rangle, \\
               \langle \Server1, 8 \rangle,\\ 		  
		       \langle \Server1, 9 \rangle,\\ 	
			       \langle \Server1, 10 \rangle,\\ 	 				
		\end{array} $\\
		\end{tabular}\\
		\\

		\framebox{\noindent\textbf{$l|_{\Notary1}$:}}\\

		\begin{tabular}{l}
		$
		\begin{array}{l}
		\langle \langle \User1, 5 \rangle,  \langle \Notary1, 1 \rangle \rangle, \\
                  \langle \langle \Notary1, 2 \rangle, \langle \User1, 8 \rangle \rangle, \\
		\end{array} $ 
		\end{tabular}\\
               \\

		\framebox{\noindent\textbf{$l|_{\Notary2}$:}}\\

		\begin{tabular}{l}
		$
		\begin{array}{l}
		\langle \langle \User1, 6 \rangle,  \langle \Notary2, 1 \rangle \rangle, \\
                	\langle \langle \Notary2, 2 \rangle, \langle \User1, 9 \rangle \rangle,  \\
		\end{array} $ 
		\end{tabular}\\
               \\
               
               \framebox{\noindent\textbf{$l|_{\Notary3}$:}}\\

		\begin{tabular}{l}
		$
		\begin{array}{l}
		\langle \langle \User1, 7 \rangle,  \langle \Notary3, 1 \rangle \rangle, \\
                	\langle \langle \Notary3, 2 \rangle, \langle \User1, 10 \rangle \rangle,  \\
		\end{array} $ 
		\end{tabular}\\
               \\

			   \framebox{\noindent\textbf{$l|_{\User2}$:}}\\

		\begin{tabular}{l}
		$
		\begin{array}{l}
		      \langle \User2, 1 \rangle, \\
				\langle  \langle \User2, 2 \rangle, \langle \Server1, 4 \rangle \rangle, \\
		\langle \langle \Server1, 6 \rangle, \langle \User2, 3 \rangle \rangle, \\

		\end{array} $ 
		\end{tabular}\\

						\end{tabular}
			
			&
							\begin{tabular}[t]{l}

\small
  
	{\noindent\textbf{(c)}}\\
							\framebox{\noindent\textbf{$a_d|_{\Adversary}$ }}\\

		\begin{tabular}{l}
		$
		\begin{array}{l}
		
		\langle \langle \User1, 2 \rangle, \langle \Adversary, 1 \rangle \rangle, \\
			\langle  \Adversary, 2 \rangle,\\
	
			\langle \langle \Adversary, 3  \rangle, \langle \User1, 3 \rangle \rangle, \\
			\langle \langle \User1, 12 \rangle, \langle \Adversary, 4 \rangle \rangle,\\
			\langle  \Adversary, 5 \rangle,\\
			\langle  \Adversary, 6 \rangle,\\
			
		       \langle \langle \Adversary, 7 \rangle,\langle \Server1, 1 \rangle \rangle,\\
			\langle \langle \Server1, 3 \rangle, \langle \Adversary, 8 \rangle \rangle, \\
			\langle  \Adversary, 9 \rangle,\\
			
      		\langle \langle \Adversary, 10  \rangle,  \langle \Server1, 7 \rangle \rangle, \\

		\end{array} $ \\
		\end{tabular}\\
		\\
		
		\framebox{\noindent\textbf{$a_d|_{\User1}$}}\\
		\begin{tabular}{l}
		$
		\begin{array}{l}
		\langle \User1, 1 \rangle\\
		\langle \langle \User1, 2 \rangle, \langle \Adversary, 1 \rangle \rangle, \\
		\langle \langle \Adversary, 3 \rangle, \langle \User1, 3 \rangle \rangle \\
		\langle \User1, 4 \rangle,\\
		\langle \langle \User1, 5 \rangle,  \langle \Notary1, 1 \rangle \rangle, \\
		\langle \langle \User1, 6 \rangle,  \langle \Notary2, 1 \rangle \rangle,\\
		\\
		   \langle \langle \Notary1, 2 \rangle, \langle \User1, 8 \rangle \rangle,  \\
                  \langle \langle \Notary2, 2 \rangle, \langle \User1, 9 \rangle \rangle,  \\
                  \\
	 		\langle \User1, 11 \rangle, \\
			\langle \langle \User1, 12 \rangle, \langle \Adversary, 4 \rangle \rangle,\\
		
			\end{array} $ \\
			\end{tabular}\\
			
\\

	\framebox{\noindent\textbf{$a_d|_{\Server1}$:}}\\
		\begin{tabular}{l}
		$
		\begin{array}{l}

		  \langle \langle \Adversary, 7 \rangle,\langle \Server1, 1 \rangle \rangle,\\
		
   				\langle \Server1, 2 \rangle \\
				\langle \langle \Server1, 3 \rangle, \langle \Adversary, 8 \rangle \rangle, \\
				\\
				\\
				\\
				
				\langle \langle \Adversary, 10 \rangle,  \langle \Server1, 7 \rangle \rangle, \\

               \langle \Server1, 8 \rangle,\\ 		  
		       \langle \Server1, 9 \rangle,\\ 	
			       \langle \Server1, 10 \rangle,\\

		\end{array} $\\
		\end{tabular}\\
			\\
		\framebox{\noindent\textbf{$a_d|_{\Notary1}$:}}\\

		\begin{tabular}{l}
		$
		\begin{array}{l}
		\langle \langle \User1, 5 \rangle,  \langle \Notary1, 1 \rangle \rangle, \\
                  \langle \langle \Notary1, 2 \rangle, \langle \User1, 8 \rangle \rangle, \\
		\end{array} $ 
		\end{tabular}\\
               \\

		\framebox{\noindent\textbf{$a_d|_{\Notary2}$:}}\\

		\begin{tabular}{l}
		$
		\begin{array}{l}
		\langle \langle \User1, 6 \rangle,  \langle \Notary2, 1 \rangle \rangle, \\
                	\langle \langle \Notary2, 2 \rangle, \langle \User1, 9 \rangle \rangle,  \\
		\end{array} $ 
		\end{tabular}\\
 
						\end{tabular}

			\end{tabular}

\caption{\emph{Left to Right:} \textbf{(a):} $\log(t)|_i$ for $i \in \{\Adversary, \User1, \Server1, \Notary1, \Notary2, \Notary3, \User2, \User3\}$. \textbf{(b):} Lamport cause $l$ for Theorem~\ref{theorem-passwords2}. $l|_i= \emptyset$ for $i \in \{\Notary4, \Server2, \User4, \User3\}$ as output by Definition~\ref{definition:cause1'}. \textbf{(c):} Actual cause $a_d$ for Theorem~\ref{theorem-passwords2}. $a_d|_i= \emptyset$ for $i \in \{\Notary3, \Notary4, \Server2, \User4, \User2, \User3\}$. $a_d$ is a projected \emph{sublog}
of Lamport cause $l$.}
\label{fig:log1}
\end{figure*}

\begin{figure}

\small
		
		\begin{tabular}{p{0.5\textwidth}}

		\framebox{\noindent\textbf{$\log(t)|_{\Server2}$:}}\\
		\begin{tabular}{l}
		$
		\begin{array}{l}
 		\langle \langle \User4, 2 \rangle, \langle \Server2, 1 \rangle \rangle, \\
		\langle \Server2, 2 \rangle, \\
		 \langle \langle \Server2, 3 \rangle, \langle \User4, 3 \rangle \rangle, \\
		 \langle \langle \User4, 8 \rangle, \langle \Server2, 4 \rangle \rangle, \\\
 		\langle \Server2, 5 \rangle, \\
		\langle \Server2, 6 \rangle, \\
		\langle \Server2, 7 \rangle, \\
		
		\end{array} $\\
		\end{tabular}\\
		\\

		\framebox{\noindent\textbf{$\log(t)|_{\User1}$}}\\
		\begin{tabular}{l}
		$
		\begin{array}{l}
		 \langle \User4, 1 \rangle, \\
		  \langle \langle \User4, 2 \rangle, \langle \Server2, 1 \rangle \rangle, \\
		 \langle \langle \Server2, 3 \rangle, \langle \User4, 3 \rangle \rangle, \\
		  \langle \User4, 4 \rangle, \\
		 \langle \langle \User4, 5 \rangle, \langle \Notary4, 1 \rangle \rangle, \\
		 \langle \langle \Notary4, 3 \rangle, \langle \User4, 6 \rangle \rangle, \\
		   \langle \User4, 7 \rangle, \\
		 \langle \langle \User4, 8 \rangle, \langle \Server2, 4 \rangle \rangle, \\\
		
			\end{array} $ \\
			\end{tabular}\\
			
\\
			
		\framebox{\noindent\textbf{$\log(t)|_{\Notary4}$:}}\\

		\begin{tabular}{l}
		$
		\begin{array}{l}
		 \langle \langle \User4, 5 \rangle, \langle \Notary4, 1 \rangle \rangle, \\
		 \langle \Notary4, 2 \rangle, \\
		 \langle \langle \Notary4, 3 \rangle, \langle \User4, 6 \rangle \rangle, \\
		\end{array} $ 
		\end{tabular}\\
               \\

			\end{tabular}

\caption{$\log(t)|_i$ where $i \in \{\User4, \Server2, \Notary4\}$}
\label{fig:log2}
\end{figure}


Note that, if any two of the three notaries
had attested the \Adversary's key to belong to \Server1, the violation would have still
happened. Consequently, we may expect three independent program
causes in this
example: \{\Adversary, \User1, \Server1, \Notary1, \Notary2\} with the
action causes $a_d$ as shown in Figure~\ref{fig:log1}(c),
\{\Adversary, \User1, \Server1, \Notary1, \Notary3\} with the actions $a_d'$, and
\{\Adversary, \User1, \Server1, \Notary2, \Notary3\} with the actions $a_d''$ where $a_d'$ and $a_d''$ can be obtained from $a_d$ (Figure~\ref{fig:log1}(c)) by considering actions for \{\Notary1, \Notary3\} and \{\Notary2, \Notary3\} respectively, instead of actions for \{\Notary1, \Notary2\}. The following
theorem states that our definitions determine exactly these three
independent causes.

\begin{theorem}\label{theorem-passwords2}
Let $I = \{\User1,\Server1,\Adversary,\Notary1,\Notary2,\Notary3,
\Notary4, \Server2, \User4, \User2, \User3\}$, and $\Sigma$ and ${\cal A}$ be as described above. Let $t$ be a trace
from $\langle I, {\cal A}, \Sigma\rangle$ such that $log(t)|_i$ for
each $i \in I$ matches the corresponding log projection from
Figures~\ref{fig:log1}(a) and~\ref{fig:log2}. Then, Definition~\ref{definition:cause2b}
determines three possible values for the program cause $X$ of
violation
$t \in \varphi_V$: \{\Adversary, \User1, \Server1, \Notary1, \Notary2\},
\{\Adversary, \User1, \Server1, \Notary1, \Notary3\}, and
\{\Adversary, \User1, \Server1, \Notary2, \Notary3\} where the corresponding actual causes are $a_d, a_d'$ and $a_d''$ respectively.

\end{theorem}

It is instructive to understand the proof of this theorem, as it
illustrates our definitions of causation.  We verify that our Phase~1 and Phase~2 definitions (Definitions~\ref{definition:cause1'}, ~\ref{definition:cause2a}, ~\ref{definition:cause2b}) yield
exactly the three values for $X$ mentioned in the theorem. 

\paragraph{Lamport cause (Phase 1)}

We show that any $l$ whose projections match those shown in
Figure~\ref{fig:log1}(b) satisfies sufficiency and minimality. From
Figure~\ref{fig:log1}(b), such an $l$ has no actions for \User3, \User4, \Notary4, \Server2 and
only those actions of \User2 that are involved in synchronization
with \Server1. For all other threads, the log contains every action
from $t$. The intuitive explanation for this $l$ is straightforward:
Since $l$ must be a (projected) \emph{prefix} of the trace, and the
violation only happens because of $\action{insert}$ in the last
statement of \Server1's program, every action of every program before
that statement in Lamport's happens-before relation must be in
$l$. This is exactly the $l$ described in Figure~\ref{fig:log1}(b).

Formally, following the statement of sufficiency, let $T$ be the set
of traces starting from ${\cal C}_0 = \langle {I}, {\cal
A}, \Sigma \rangle$ (Figure~\ref{fig:actuals1}) whose logs contain
$l$ as a projected prefix. Pick any $t' \in T$. We need to show
$t' \in \varphi_V$. However, note that any $t'$ containing all actions
in $l$ must also add $(acct_1, \Adversary)$ to $P_1$, but
$\Adversary \not= \User1$. Hence, $t' \in \varphi_V$. Further, $l$ is
minimal as described in the previous paragraph.

\paragraph{Actual cause (Phase 2)}
Phase 2 (Definitions~\ref{definition:cause2a}, ~\ref{definition:cause2b}) determines
three independent program causes for
$X$: \{\Adversary, \User1, \Server1, \Notary1, \Notary2\}, \{\Adversary, \User1, \Server1, \Notary1, \Notary3\},
and \{\Adversary, \User1, \Server1, \Notary2, \Notary3\} with the actual action causes given by $a_d, a_d'$ and $a_d''$, respectively in Figure~\ref{fig:log1}(c). These are
symmetric, so we only explain why $a_d$ satisfies Definition~\ref{definition:cause2a}. (For this $a_d$,
Definition~\ref{definition:cause2b} immediately forces $X
= \{\Adversary, \User1, \Server1, \Notary1, \Notary2\}$.)  We show that (a) $a_d$ satisfies sufficiency', and (b) No proper \emph{sublog} of $a_d$ satisfies
sufficiency' (minimality'). Note that $a_d$ is obtained from
$l$ by dropping \Notary3, \User2 and \User3, and all their
interactions with other threads.

We start with (a). Let $a_d$ be such that $a_{d}|_i$ matches
Figure~\ref{fig:log1}(c) for every $i$. Fix any dummifying function
$f$. We must show that any trace originating from ${\sf dummify}(I,
{\cal A}, \Sigma,a_d,f)$, whose log contains $a_d$ as a projected
sublog, is in $\varphi_V$. Additionally we must show that there is
such a trace. There are two potential issues in mimicking the
execution in $a_d$ starting from ${\sf dummify}(I, {\cal
A}, \Sigma,a_d,f)$ --- first, with the interaction between \User1
and \Notary3 and, second, with the interaction between \Server1
and \User2. For the first interaction, on line~7, ${\cal A}(\User1)$
(Figure~\ref{fig:actuals1}) synchronizes with {\Notary3} according to
$l$, but the synchronization label does not exist in $a_d$. However,
in ${\sf dummify}(I, {\cal A}, \Sigma,a_d,f)$, the $\action{recv}()$
on line~10 in ${\cal A}(\User1)$ is replaced with a dummy value, so the
execution from ${\sf dummify}(I, {\cal A}, \Sigma,a_d,f)$
progresses. Subsequently, the majority check (assertion [B]) succeeds
as in $l$, because two of the three notaries ({\Notary1} and
{\Notary2}) still attest the {\Adversary}'s key. 

A similar observation can be made about the interaction between \Server1 and \User2.  Line~4,
${\cal A}(\Server1)$ (from
Figure~\ref{fig:log1}(b)) synchronizes with {\User2} according to $l$,
but this synchronization label does not exist in $a_d$. However, in ${\sf dummify}(I, {\cal A}, \Sigma,a_d,f)$, the
$\action{recv}()$ on line~4 in ${\cal A}(\Server1)$ is replaced with a
dummy value, so the execution from ${\sf dummify}(I, {\cal A}, \Sigma,a_d,f)$
progresses. Subsequently, \Server1 still adds permission for the \Adversary.

Next we prove that every trace starting from ${\sf dummify}(I, {\cal A}, \Sigma,a_d,f)$, whose log contains $a_d$ (Figure~\ref{fig:log1}(c)) as a projected \emph{sublog}, is in $\varphi_V$. Fix a trace $t'$ with log $l'$. Assume $l'$
coincides with $a_d$. We show $t' \in \varphi_V$ as
follows:
\begin{enumerate}
\item 
Since the synchronization labels in $l'$ are a superset of those in
$a_d$, {\Server1} must execute line~10 of its program ${\cal
A}(\Server1)$ in $t'$. After this
line, the access control matrix $P_1$ contains $(acct_1, J)$ for some
$J$.
\item
When ${\cal A}(\Server1)$ writes $(x, J)$ to $P_1$ at line~10, then $J$
is the third component of a tuple obtained by decrypting a message
received on line~7.
\item
Since the synchronization projections on $l'$ are a superset of $a_d$,
and on $a_d$ $\langle \Server1, 7\rangle$ synchronizes with
$\langle \Adversary, 10\rangle$, $J$ must be the third component of an
encrypted message sent on line~10 of ${\cal A}(\Adversary)$.

\item 
The third component of the message sent on line~10 by {\Adversary} is
exactly the term ``\Adversary''. (This is easy to see, as the term
``\Adversary'' is hardcoded on line 9.) Hence, $J = \Adversary$.
\item 
This immediately implies that $t' \in \varphi_V$ since
$(acct_1, \Adversary) \in P_1$, but ${\Adversary} \not= {\User1}$.
\end{enumerate}

Last, we prove (b) --- that no proper subsequence of $a_d$ satisfies
sufficiency'. Note that $a_d$ (Figure~\ref{fig:log1}(c)) contains
exactly those actions from $l$ (Figure~\ref{fig:log1}) on whose
returned values the last statement of \Server1's program
(Figure~\ref{fig:actuals1}) is data or control
dependent. Consequently, all of $a_d$ as shown is necessary to obtain
the violation.

In particular, observe that if labels for \Server1 ($a_d|_{\Server1}$) are not a part of $a_d'$, then {\Server1}'s labels are not
in ${\sf dummify}(I, {\cal A}, \Sigma,a_d,f)$ and, hence, on any counterfactual trace
\Server1 cannot write to $P_1$, thus precluding a violation. Therefore,
the sequence of labels in $a_d|_{\Server1}$ are required in the actual cause. 

By sufficiency', for any $f$, the log of trace $t'$ of ${\sf dummify}(I, {\cal A}, \Sigma,a_d,f)$ must contain $a_d$ as a projected \emph{sublog}. This means that in $t'$, the
assertion [A] of ${\cal A}(\Server1)$ must succeed and, hence, on
line 7, the correct password $pwd_1$ must be received by {\Server1},
independent of $f$. This immediately implies that {\Adversary}'s action of sending that password must be in $a_d$, else some dummified
executions will have the wrong password sent to \Server1 and the assertion [A] will fail. 

Extending this logic further, we now observe that because {\Adversary}
forwards a password received from {\User1} (line 4 of ${\cal
A}(\Adversary)$) to {\Server1}, the send action of $\User1$ will be in $a_d$ (otherwise, some
dummifications of line 4 of ${\cal A}(\Adversary)$ will result in the
wrong password being sent to {\Server1}, a contradiction).  Since
$\User1$'s action is in $a_d$ and $l'$ must contain $a_d$ as a \emph{sublog}, the majority check of ${\cal A}(\User1)$ must also
succeed. This means that at least two of
$\{\Notary1, \Notary2, \Notary3\}$ must send the confirmation to \User1, else the
dummification of lines 8 -- 10 of ${\cal N}(\User1)$ will cause the
assertion [B] to fail for some $f$. Since we are looking for a minimal \emph{sublog} therefore we only consider the send actions from two threads i.e. $\{\Notary1, \Notary2\}$. At this point we
have established that each of the labels  as shown in Figure~\ref{fig:log1}(c) are required in
$a_d$. Hence, $a_d' = a_d$.

\end{document}